%% file: main.tex
\documentclass[lettersize,journal]{IEEEtran}
\usepackage{cite}
\usepackage{amsmath,amssymb,amsfonts}
\usepackage{amsthm}
\usepackage{graphicx}
\usepackage{textcomp}
\usepackage{xcolor}
\def\BibTeX{{\rm B\kern-.05em{\sc i\kern-.025em b}\kern-.08em
    T\kern-.1667em\lower.7ex\hbox{E}\kern-.125emX}}

\usepackage{xcolor}
\usepackage{colortbl}
\usepackage{url}
\usepackage{xurl}
\usepackage{booktabs}
\usepackage[ruled,vlined]{algorithm2e}
\usepackage{algpseudocode}
\usepackage{enumerate}
\usepackage[english]{babel}
\usepackage{blindtext}
\usepackage{amsmath}

\usepackage{amssymb}
\usepackage{xspace}
\usepackage{multirow}
\usepackage{threeparttable}
\usepackage{float}
\usepackage{flafter}
\usepackage{comment}
\usepackage{enumitem}
\usepackage{subfig}
\usepackage{graphicx}
\usepackage{lipsum}

\def\fig{Fig.\xspace}
\def\eqn{Eq.\xspace}
\def\sec{Sec.\xspace}
\def\tab{Tab.\xspace}

\def\eg{{\textit{e.g.}\xspace}}

\def\etc{{\textit{etc}\xspace}}

\newtheorem{proposition}{Proposition}

\SetKwInput{KwIn}{Input}
\SetKwInput{KwOut}{Output}
\SetKwInput{KwResult}{Result}

\newcommand{\head}[1]{{\noindent \textbf{#1:}}}
\newcommand{\term}[1]{{\textit{#1}}}

\graphicspath{{./figures/}}
\graphicspath{{./pdf/}}
\DeclareGraphicsExtensions{.pdf,.jpeg,.png}

\def\sysname{\textsc{WuKong}\xspace}

\def\networkname{FreDiT\xspace}

\ifodd 1
\newcommand{\com}[1]{\textbf{\color{red}(COMMENT: #1)}} %comment of the text
\newcommand{\todo}[1]{\textbf{{\color{orange}(TODO: #1)}}}
\newcommand{\unused}[1]{{\color{gray}#1}}
\newcommand{\sheng}[1]{\textbf{\color{olive}(Sheng: #1)}} %comment of the text
\else

\newcommand{\com}[1]{}
\newcommand{\todo}[1]{}
\newcommand{\unused}[1]{}
\newcommand{\sheng}[1]{}

\fi

\begin{document}

\title{
``Buy One Get N Free": Neuro-Wideband WiFi Sensing via Self-Conditioned CSI Extrapolation
}

\author{Sijie Ji,~\IEEEmembership{Member,~IEEE}, Weiying Hou,~\IEEEmembership{Student Member,~IEEE}, Chenshu Wu,~\IEEEmembership{Senior Member,~IEEE }

\IEEEauthorblockA{Department of Computer Science,  The University of Hong Kong
}
\\sijieji@caltech.edu, chenshu@cs.hku.hk
}

\markboth{Journal on Selected Areas in Communications,~Vol.~01, No.~6, August~2025}%
{Shell \MakeLowercase{\textit{et al.}}: A Sample Article Using IEEEtran.cls for IEEE Journals}


\maketitle

\begin{abstract}
WiFi sensing has suffered from the limited bandwidths designated for its original communication purpose, leading to fundamental limits in multipath resolution and thus multi-user sensing.
Unfortunately, it is practically prohibitive to obtain large bandwidths on commercial WiFi, considering the conflict between the limited spectrum and the crowded networks. 
In this paper, we present \term{Neuro-Wideband} (NWB), a completely different paradigm that enables wideband WiFi sensing without specialized hardware or extra channel measurements. 
Our key insight is that any physical measurement of channel state information (CSI) inherently encapsulates multipath parameters, which, while unsolvable in isolation, can be transformed into an expanded form of CSI (\term{eCSI}) approximating measurements over a broader bandwidth. To ground this insight,  
we propose \sysname to address NWB as a unique self-conditioned learning problem that can be trained by using any existing CSI data as self-labeled samples. 
\sysname introduces a novel deep learning framework by integrating Transformer and Diffusion models, which captures sample-specific multipath parameters and transfers this sample-level knowledge to the outcome eCSI. 
We conduct real-world experiments to evaluate \sysname on diverse WiFi signals across protocols and bandwidths. The results show the promising effectiveness of NWB, which is further demonstrated through case studies on localization and multi-person breathing monitoring using eCSI. 
Overall, the proposed NWB promises a practical pathway toward realizing wideband WiFi sensing on commodity hardware, expanding the design space of wireless sensing systems.
\end{abstract}

\begin{IEEEkeywords}
WiFi sensing, CSI, learning-enabled wireless sensing, commodity WiFi systems, multipath resolution
\end{IEEEkeywords}

\input{body}

\bibliographystyle{IEEEtran}
\bibliography{refs}

 

% \end{IEEEbiography}



\end{document}

%% file: body.tex
%!TEX root = ./main.tex

\begin{figure}[t]
  \centering
  \includegraphics[width=0.9\linewidth]{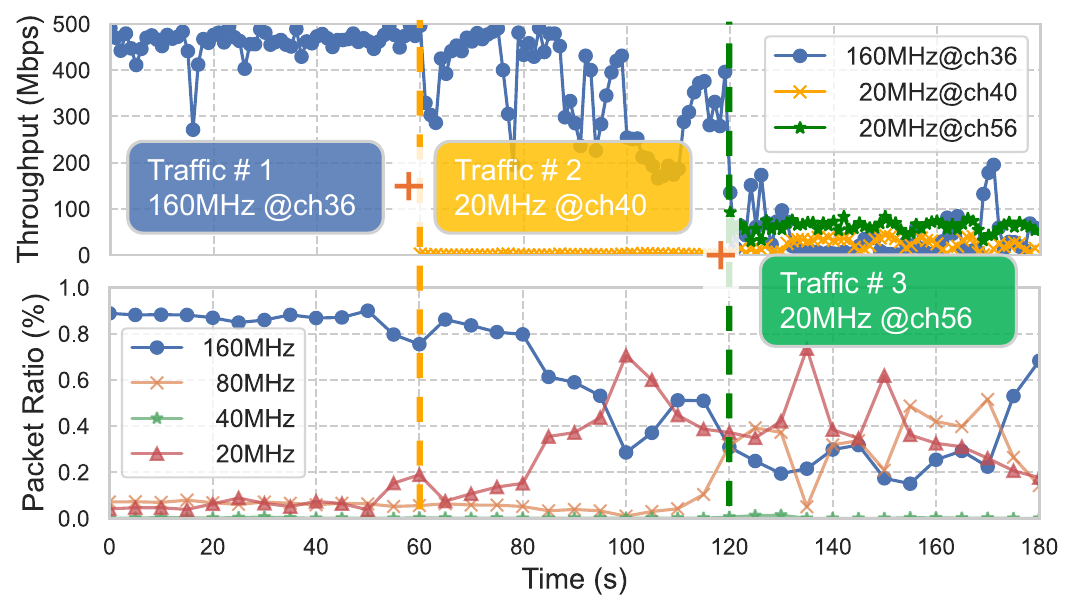}
  \vspace{-0.9\baselineskip}
  \caption{Limited WiFi bandwidth: 160MHz traffic (ch36) will degrade to 80/40/20 MHz bandwidths due to interference from overlapping 20MHz channels (ch40, ch56).}
 %\vspace{-3mm}
  \label{fig:interference}
\end{figure}

\section{Introduction}
\label{sec:intro}

Over the past decades, Radio Frequency (RF) technology has expanded far beyond its original purpose of communication, enabling a wide array of sensing applications. However, today’s widely deployed WiFi networks typically operate with relatively narrowband signals, constrained by requirements such as spectral efficiency, interference management, power consumption, and regulatory compliance. While sufficient for communication, these bandwidth limitations pose fundamental challenges for RF sensing—directly impacting range resolution and the ability to resolve multipath.

Even with modern WiFi standards like IEEE 802.11ac/ax supporting up to 160MHz bandwidth, the time resolution is limited to 6.25 ns, corresponding to a coarse range resolution of 1.875 meters. In practice, real-world deployments rarely utilize 160MHz channels: our measurements on campus WiFi networks (see \fig\ref{fig:interference}) show that fewer than 1\% of public access points operate at 160MHz\footnote{We sniffed beacon frames in multiple scenarios and observed <1\% of APs using 160MHz.}, due to spectrum congestion and dynamic fallback to narrower bands. As a result, most commodity WiFi devices operate in the 20–80MHz range, offering significantly poorer sensing resolution.

While WiFi sensing has achieved remarkable progress, bandwidth remains a fundamental bottleneck—limiting fine-grained tasks like precise localization and multi-target tracking.

Several approaches have been proposed to address this limitation. A direct method is to perform frequency hopping across adjacent channels and stitch the resulting Channel State Information (CSI) to approximate wideband measurements~\cite{vasisht2016decimeter,xie2015precise,xiong2015tonetrack,li2024uwb}. However, frequency hopping is practically infeasible: it introduces latency, fails to guarantee channel continuity, and often violates the coherence time constraint—especially on commodity WiFi devices (detailed in \S\ref{sec:background}). While future WiFi standards may support broader bands, actual bandwidth availability remains highly questionable given the limited spectrum between 2.4–7.1GHz and increasingly dense deployments.

Recent work has explored using deep learning to synthesize or augment CSI~\cite{chi2024rf,zhao2023nerf,li2024uwb}. These methods can enrich training data but typically do not produce physically consistent or continuous wideband CSI. For instance, UWB-FI~\cite{li2024uwb} samples CSI across non-contiguous bands via hardware hopping and learns AoA–ToF tuples from them, but the resulting data lacks continuity and often requires specialized devices.

To the best of our knowledge, no prior work can produce truly \textit{continuous, physically meaningful wideband CSI} on commodity WiFi hardware.

In this paper, we propose a fundamentally different paradigm, called \textbf{Neuro-Wideband (NWB)}, that enables the extrapolation of wideband CSI (\textbf{eCSI}) from a single measurement on a standard narrowband channel—\textit{without} hardware modification, additional probing, or inter-channel stitching.

Our core insight builds on the observation that at a fixed time and location, the underlying multipath parameters are physically constant across frequencies. While difficult to extract directly, these parameters leave implicit signatures in the CSI measured at any frequency. Therefore, given a narrowband CSI sample, it should be possible to infer CSI at unmeasured frequencies by learning how these underlying parameters manifest across frequency variations. In effect, NWB extrapolates the CSI across frequencies by transferring this latent multipath structure.

We present \textbf{\sysname}, the first system that realizes the NWB paradigm through a carefully designed deep learning framework. \sysname overcomes multiple technical challenges to generate high-fidelity eCSI from any narrowband input. It introduces the following key innovations:

\begin{itemize}
  \item \textbf{Self-conditioned Extrapolation.} We formulate NWB as a self-conditioned learning problem, where CSI at arbitrary input bandwidths (e.g., 20/40/80MHz) is used to predict CSI at larger target bandwidths. Our formulation ensures that eCSI remains physically meaningful by conditioning predictions on real measured data.

  \item \textbf{Self-supervised Training without Wideband Labels.} \sysname eliminates the need for wideband ground-truth labels by using random sub-bands of existing CSI as inputs and the full band as the supervision. This allows training on any available CSI data—public or in-house—without requiring wideband measurements.

  \item \textbf{Sample-specific Multipath Inference.} Unlike typical deep learning tasks that learn population-level statistics, NWB aims to infer sample-specific multipath structure. To solve this, \sysname integrates a frequency-aware Transformer with self-conditioned diffusion modeling. This architecture allows the model to internalize multipath patterns and transfer them across frequency bands.

\end{itemize}

\sysname offers several desirable properties: it handles arbitrary CSI input lengths and bandwidths; generalizes across hardware and protocols; produces seamless, continuous eCSI; and adapts to complex multipath conditions in dynamic environments.

We fully implement and evaluate \sysname using both public and self-collected datasets across diverse scenarios. Our results show that \sysname achieves high-fidelity eCSI extrapolation (e.g., from 20→160MHz), with strong alignment to measured 160MHz ground-truth CSI. Moreover, case studies on localization and multi-user breathing estimation demonstrate the utility of NWB sensing in real-world applications. We believe NWB opens up a promising direction for bandwidth-limited RF sensing, extending the capabilities of commodity WiFi through frequency extrapolation.

Our contributions are summarized as follows:
\begin{itemize}
    \item We introduce \term{Neuro-Wideband}  (NWB) for the first time, an effective and promising paradigm towards ultrawideband sensing using commodity WiFi signals, with zero extra costs in terms of additional channel measurements or specialized hardware. 
    \item We formulate NWB as a self-conditioned deep learning problem that extrapolates a physical CSI measurement of narrow bandwidth to an extended eCSI of broader bandwidth, ensuring the expanded eCSI is physically evidenced by the input and allowing any existing CSI data to serve as self-labeled training samples. 
    \item We present \sysname with a novel deep learning framework by integrating Transformer with Diffusion models, which can transfer sample-level self-conditioned knowledge to the learning outcome, setting it apart from conventional learning networks that learn the overall statistical distribution of a given dataset. 
    \item We conduct real-world experiments to evaluate \sysname across multiple WiFi protocols, bandwidths, and hardware platforms, and demonstrate the effectiveness of NWB through case studies on eCSI sensing for localization and multi-person breathing monitoring. 
\end{itemize}

\section{Background and Context}
\label{sec:background}

\subsection{The Wireless Channel and CSI}
Define $\Omega$ as the space of multi-path environmental parameters, which describes how the signal changes as it propagates from transmitter to receiver. $\boldsymbol\Omega$ is characterized by various physical phenomena such as reflection, refraction, diffraction, scattering, \etc, in addition to hardware factors. It contains all the parameters of physical environment that determine the final CSI, such as multipath number $L$, path gain $\boldsymbol{\alpha}$, path delay $\boldsymbol{\tau}$, path angle $\boldsymbol{\theta}$. Define $\mathcal F$ as the space of frequency, which includes continuous frequency points within the frequency band of interest. CSI is determined by physical channel information $\boldsymbol\Omega$ and the transmission frequency $\mathcal F$, which can be represented by the following function mapping:
\begin{equation}
\mathcal H:\boldsymbol\Omega \times \mathcal F\rightarrow \boldsymbol H.
\end{equation}
Given the physical parameters $\Omega_i$$=\left\{L,\left\{\tau_l,\alpha_l,\theta_l\right\}\right\}$$\in\boldsymbol\Omega$ of the environment and a specific frequency $f_i\in\mathcal F$,
the wireless channel can be expressed as $\boldsymbol H=\mathcal H(\Omega_i,f_i)$,
which is given by
\begin{equation}
\mathcal{H}(\Omega_i,f_i)=\sum_{\ell=1}^L\left|\alpha_{\ell}\right| e^{j \angle\alpha_{\ell}} e^{-j 2 \pi f_i\tau_{\ell}}.
\end{equation}
Considering a uniform linear array with half-wavelength inter-antenna spacing, the received channel on the $n$-th antenna can be denoted as:
\begin{equation}
\label{eqn:csi}
\mathcal{H}_{n}(\Omega_i,f_i)=\sum_{\ell=1}^L
(\left|\alpha_{\ell}\right| e^{j \angle\alpha_{\ell}} e^{-j 2 \pi f_i \tau_{\ell}}) e^{-j 2 \pi \frac{n\cos\theta_{\ell}}{2}}.
\end{equation}

\subsection{Limitations of Existing Arts}
Let $\boldsymbol H^s=\left\{\mathcal H (\Omega_i,f_{i,b}),\forall f_{i,b}\in\Pi(f_i)\right\}$ denote the sampled narrowband CSI, representing channel data observed at frequency points $\Pi(f_i)$ within the specified narrow bandwidth. Additionally, let $\boldsymbol H^e=\left\{\mathcal H (\Omega_i,f_{i,b}),\forall f_{i,b}\in\Pi^e(f_i,k)\right\}$ denote the extended \term{eCSI} over a broader bandwidth, containing channel data extrapolated by a factor of $k$ from the measured narrowband CSI $\boldsymbol H^s$  at the same measurement time. Here, $\Pi(f_i)$ represents the continuous subcarriers of the narrowband CSI $\boldsymbol H^s$ centered at frequency $f_i$, while $\Pi^e(f_i,k)$ denotes the extended subcarriers of $\Pi(f_i)$ by a factor of $k$, also centered at $f_i$.
The goal of \sysname is to leverage the current physical CSI, $\boldsymbol{H}^s$, as a condition to learn an extrapolation function such that:
\begin{equation}
 \Psi:\boldsymbol H^s\rightarrow \boldsymbol H^e.
\end{equation}
The extrapolation function $\Psi$ ultimately outputs an \term{eCSI}, which features a $k$ times wider bandwidth across a continuous frequency band $\Pi^e(f_i,k)$, compared with the original CSI. 
\sysname is designed to take physical CSI of arbitrary bandwidths as input and, in principle, supports any extrapolation factor $k$. 
We expect the obtained \term{eCSI} to be continuous across the target band as the prediction is done as a whole over all the entire frequency band rather than on individual subcarriers. 
The \term{eCSI} shall also be adaptive to and independent of diverse impacting factors such as environments, locations, devices, \etc, since each \term{eCSI} prediction is conditioned on the physical CSI that has already captured these domain-specific parameters in $\Omega_i$.

It is important to note that our goal towards NWB differs significantly from previous efforts for achieving wideband WiFi sensing, as detailed in the following:

\head{(1) Channel Hopping}
Channel hopping aims to swiftly scan through different channels and stitch the CSI measurements across channels to form a wideband WiFi signal such that: 
$$\left\{\mathcal{H}\left(\Omega_1,f_1 \right),... , \mathcal{H}\left(\Omega_n,f_n\right)\right\} ,$$ where $n$ is the total number of the scanned channels. 
Channel hopping faces two major issues. 
First, stitching the multi-channel measurements is extremely challenging due to the non-contiguous nature of different frequency bands and the intricate different sources of errors~\cite{vasisht2016decimeter,xiong2015tonetrack}, particularly because the measurements have to be captured at different times, across which the channel may have already changed. This results in the physical environment characteristics corresponding to the channels from different frequency bands being different, i.e., $\Omega$ varies across different frequency bands.
Second, channel hopping methods inevitably impose extra overhead that would impacting normal communication~\cite{xie2019md}. 

\head{(2) Channel Mapping}
Channel mapping in wireless communications involves predicting the CSI at one frequency  based on the CSI from another frequency such that:
$$\mathcal{H}\left(\Omega_i,f_i \right) \rightarrow \mathcal{H}\left(\Omega_i,f_j\right).$$
This process is crucial especially for frequency division duplex (FDD) communication systems, where the uplink and downlink channel no longer exhibit reciprocity~\cite{guo2024deep,liu_fire_2021,vasisht_eliminating_2016}. Techniques for channel mapping often leverage statistical models, machine learning algorithms, or data-driven approaches to predict how the CSI changes over fixed frequency interval. 

Regretablly, existing channel mapping solutions do not address the problem of NWB. 
First, notable efforts, such as OptML~\cite{bakshi2019fast} and R2F2~\cite{vasisht_eliminating_2016}, attempt to estimate complicated multipath parameters (\eg, angle of arrival, time of flight, attenuation, \etc, as in \eqn\eqref{eqn:csi}) for prediction, which is extremely challenging, if possible, due to the poor multipath resolutions and the highly unpredictable multipath propagations. Rather, \sysname aims to transfer the knowledge of the multipath parameters, rather than resolving them. Second,the channel mapping solutions are not physics-informed, facing challenges in adapting to different physical environments. Third, existing channel mapping is only applicable to predict channel at fixed frequency interval and cannot simultaneously predict channels over a continuous frequency band. 

\head{(3) Generative AI}
Recently, generative AI approaches~\cite{zhao2023nerf,liu_fire_2021,chi2024rf} have been explored to generate CSI samples by learning from an underlying distribution of a given dataset, which can be formulated as:
$$\left\{\mathcal H(\Omega_n,f_n)\right\}^{N} \rightarrow \left\{\boldsymbol\Omega,\mathcal F \right\}\rightarrow  \left\{\mathcal{H} (\Omega_i,f_i)\right\},$$
where $N$ is the number of CSI samples used for training. 

Generative approaches circumvent the intractable parameter estimation by using data to construct a latent space to implicitly represent the parameter distributions, with the belief that the latent space is inclusive enough for generating multipath profiles. 
The biggest drawback of generative approaches is that the network can only generate channel data that matches the statistical characteristics of the training samples. When the training samples contain channel information from a specific frequency band, the generative network is restricted to producing CSI data corresponding to that frequency band, without the ability to generalize to a wider frequency band for NWB sensing. Furthermore, the generative network can only generate data with the same physical properties as the training samples, making it difficult to adapt to different physical environments.

Overall, as far as we are aware, no prior works can obtain continuous wideband CSI without using specialized hardware or performing extra multi-channel measurements. 
\sysname achieves this by presenting a novel paradigm of NWB and formulating it as a self-conditioned learning problem that extrapolates a physical CSI measurement to a novel form of \term{eCSI} with a broader continuous bandwidth.

\subsection{Rationale of NWB}
\label{sec:overview}

Recalling \eqn\eqref{eqn:csi}, it is extremely difficult to resolve the parameters of all the multipath components, especially considering the highly unpredictable environmental dynamics and diverse hardware noises. 
While remarkable advances have been achieved in wireless sensing without large bandwidths to accurately resolve the multipath parameters, \sysname takes a fundamentally different approach towards NWB sensing.
\fig~\ref{f:comparison} illustrates the high-level idea. Unlike previous works that attempt to estimate the parameters for individual CSI entries or learn a representative model to approximate them, \sysname leverages the unknown physical parameter set associated with each CSI entry as an underlying constraint for extrapolating CSI to a larger bandwidth. 
The opportunity lies in the fact that, at any given time moment, the measured CSI $\boldsymbol{H}$ at one frequency implies a specific unknown multipath parameter set $\Omega_{\boldsymbol{H}}$, which is in principle shared by channels of different frequencies, within the same space and at the same time, and thus can be employed as an unknown yet fixed condition transferable across different subcarriers for forming $\boldsymbol{H}^e$.

%\subsection{Feasibility of \sysname}
Although CSI is derived from \eqn\eqref{eqn:csi} for a specific frequency, the multipath parameters are frequency-independent. Thanks to the Orthogonal Frequency Division Multiplexing (OFDM) modulation in WiFi, at any time, the observed CSI $\boldsymbol{H}^s$ at different carrier frequencies shares the same multipath parameters, $\boldsymbol{\Omega}$.
Thus, each CSI measurement provides diverse frequency observations with the same multipath information. Therefore, through many CSI samples from different conditions, it is possible to map across different frequencies and further model the self-conditioned CSI extrapolation. 
\begin{proposition}
\label{pos:1}
Given the narrowband CSI $\boldsymbol{H}^s$ observed over a narrow bandwidth $\Pi(f_i)$, and a continuous frequency band $\Pi^e(f_i,k)$ expanded by a factor of $k$, there exists a deterministic mapping for CSI extrapolation:
\begin{equation}
\label{eqn:phi}
\begin{aligned}
\Phi_{\Pi(f_i) \rightarrow \Pi^e(f_i, k)}:
\boldsymbol{H}^s &= \left\{ \mathcal{H} (\Omega_i, f_{i,b}), \forall f_{i,b} \in \Pi(f_i) \right\} \\
\rightarrow \boldsymbol{H}^e &= \left\{ \mathcal{H}(\Omega_i, f_{i,b}), \forall f_{i,b} \in \Pi^e(f_i, k) \right\},
\end{aligned}
\end{equation}

that is uniquely determined by the parameters of the underlying physical environment $\Omega_i$ (including the geometry, materials, antenna positions, etc.).
\end{proposition}
\begin{proof}
Based on the \eqn\eqref{eqn:csi}, the distinct physical parameters, such as ($\alpha_\ell, \tau_\ell, \theta_\ell $), are sufficient to fully define the wireless channels on any given frequency $f_i$ for each signal propagation path. None of these parameters depend on the channel frequency. At a specific time, the propagation environment remains unchanged across the continuous frequency bands $\Pi$ and $Pi^e$, which means the mapping $\Phi_{\Pi(f_i) \rightarrow \Pi^e(f_i, k)}$ is exclusively determined by the specific physical parameters $\Omega_i$. Different physical environments lead to different mapping functions.
\end{proof}
\begin{figure}[t]
       \subfloat[Exsiting practice]{
  \includegraphics[width=0.46\linewidth]{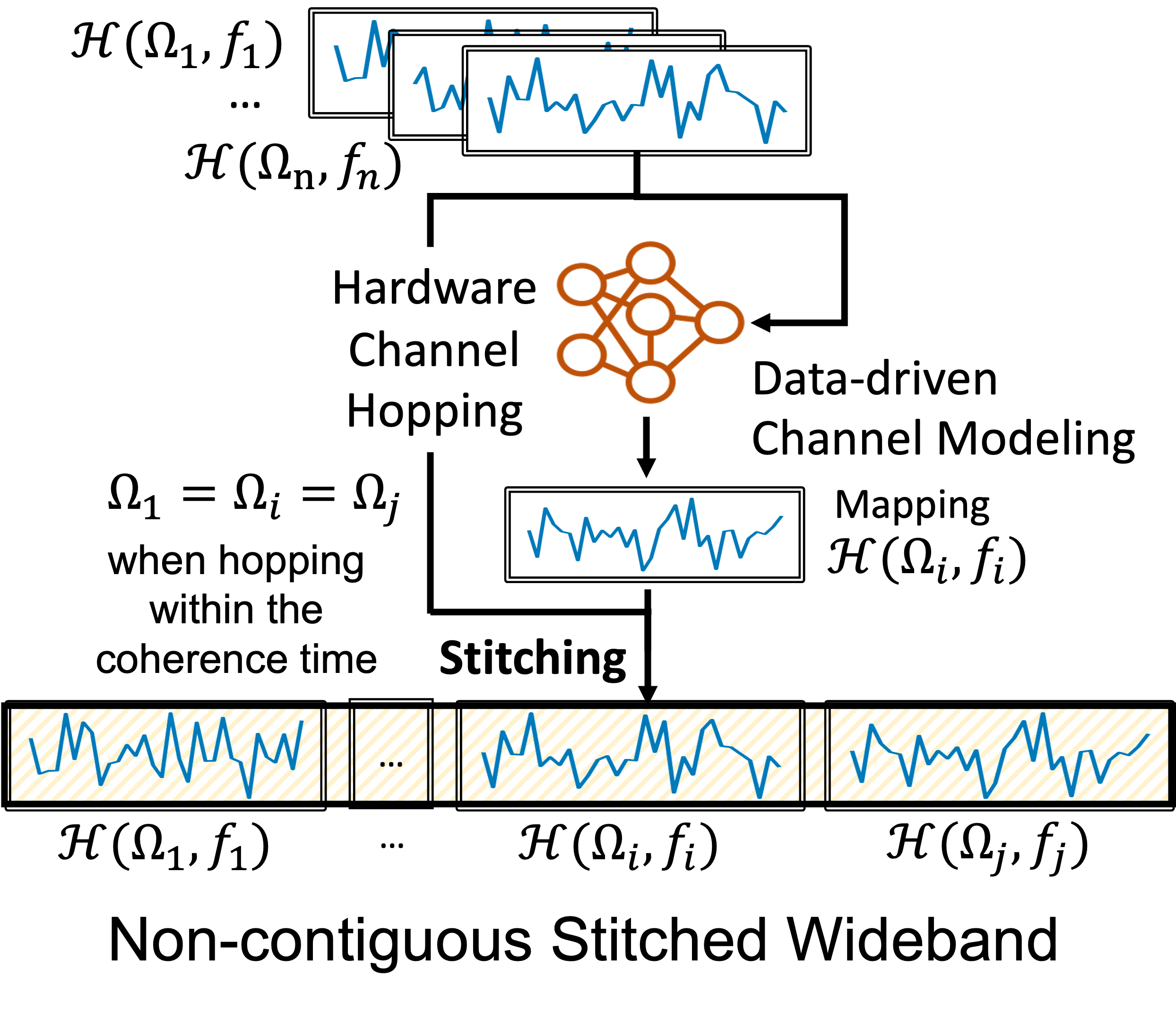}}
  \quad
    \subfloat[\sysname]{
    \label{subfig:ours}
  \includegraphics[width=0.46\linewidth]{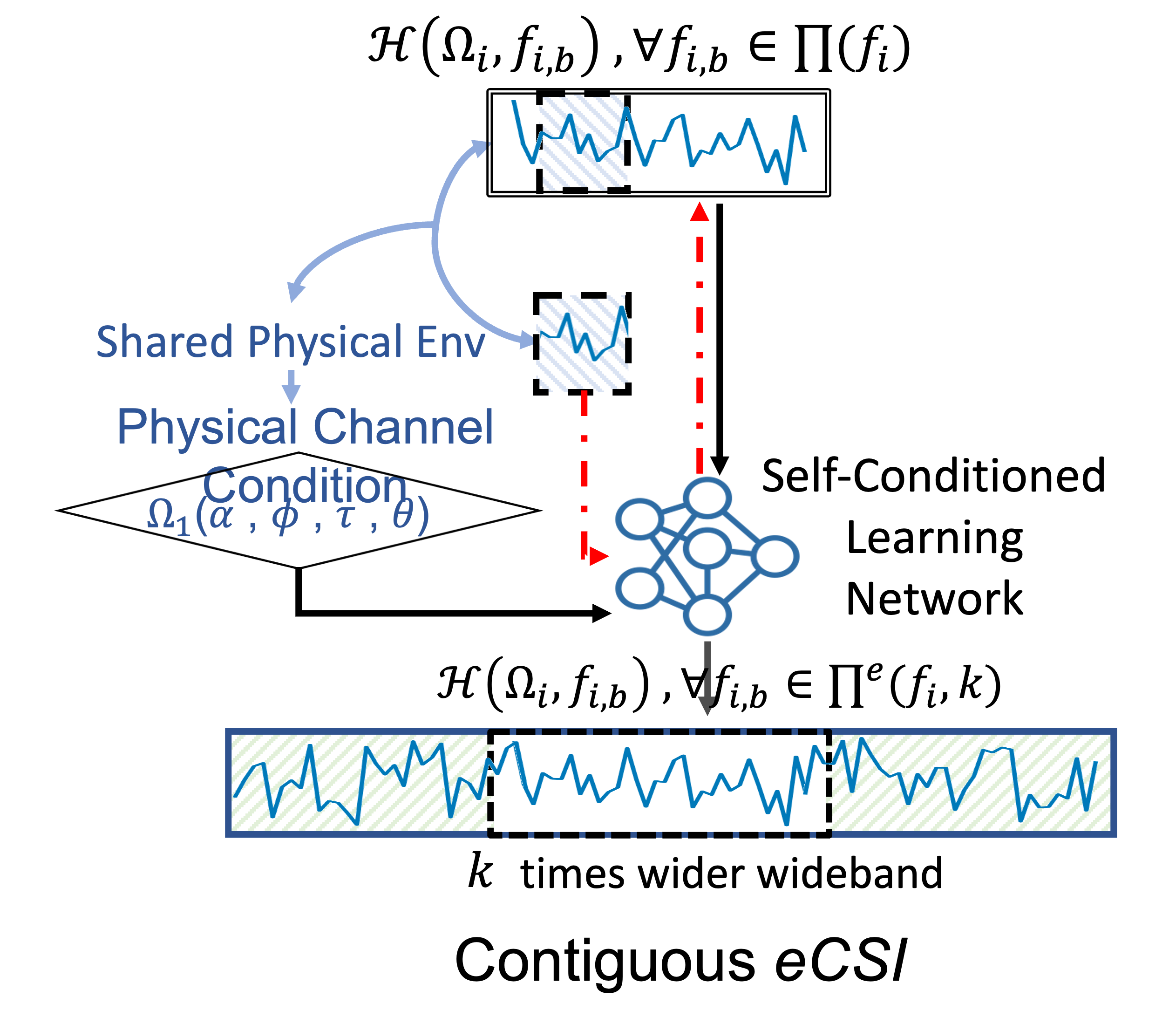}}
  \quad
   % \vspace{-0.9\baselineskip}
    \caption{(i) Existing approaches assume a determined multipath mapping from current(extensive) data samples, whereas \sysname learns to extrapolate the multipath information inherent in the sub-samples of each sample (self-conditioned). (ii) Unlike conventional methods rely on stitching to achieve large bandwidth, \sysname attains substantial bandwidth directly.}\label{f:comparison}
     \vspace{-3mm}
\end{figure}
As shown by the proposition above, the CSI extrapolation mapping function is closely related to the physical parameters of the channel. Leveraging the deep learning method to capture these underlying physical parameters from the sampled channel, and subsequently inferring the mapping function based on the physical characteristics, it is possible to successfully construct the desired mapping. The whole mapping learning process can be characterized by the following probability model:
\begin{equation}
\boldsymbol{H}^s\xrightarrow[\mathbb P(\Omega_i |\boldsymbol H^s)]{} \Omega_i\xrightarrow[\mathbb P(\boldsymbol H^e |\Omega_i)]{} \boldsymbol{H}^e,
\end{equation}
where $\mathbb P(\Omega_i |\boldsymbol H^s)$ denotes the process of physical parameter extraction from the measured narrowband CSI, $\mathbb P(\boldsymbol H^e |\Omega_i)$ denotes the process of broadband \term{eCSI} construction from  the underlying multipath parameters. The whole process of NWB can be viewed as the calculation of the following posterior distribution: 
\begin{equation}
\mathbb P(\boldsymbol H^e |\boldsymbol H^s)=\int_{\Omega_i} \mathbb P(\boldsymbol H^e|\Omega_i)\mathbb P (\Omega_i|\boldsymbol H^s),
\end{equation}
which represents the belief about the channel response at the unobserved frequencies conditioned on the measured CSI.

\head{Summary} \sysname aims to derive a deep learning function $\mathcal{G}$, which takes the current CSI, $\boldsymbol{H}^s$, as a condition to extrapolate the channel response at other unmeasured frequencies, $\boldsymbol{H}^e$, such that:
\begin{align*}
\text{argmin}_{\Theta} \quad & \mathbb E \left[ \left\|\mathcal{G}(\boldsymbol{H}^s, \boldsymbol{\Theta})- \Phi_{\Pi(f_i) \rightarrow \Pi^e(f_i, k)}(\boldsymbol{H}^s)\right\|^2\right] \\
\text{subject to} \quad & \Phi_{\Pi(f_i) \rightarrow \Pi^e(f_i, k)} \mathrm{\,in\,\eqn\eqref{eqn:phi}}, \forall\Omega,k.
\end{align*}

\begin{figure*}[t!]
  \centering
  \includegraphics[width=0.9\linewidth]{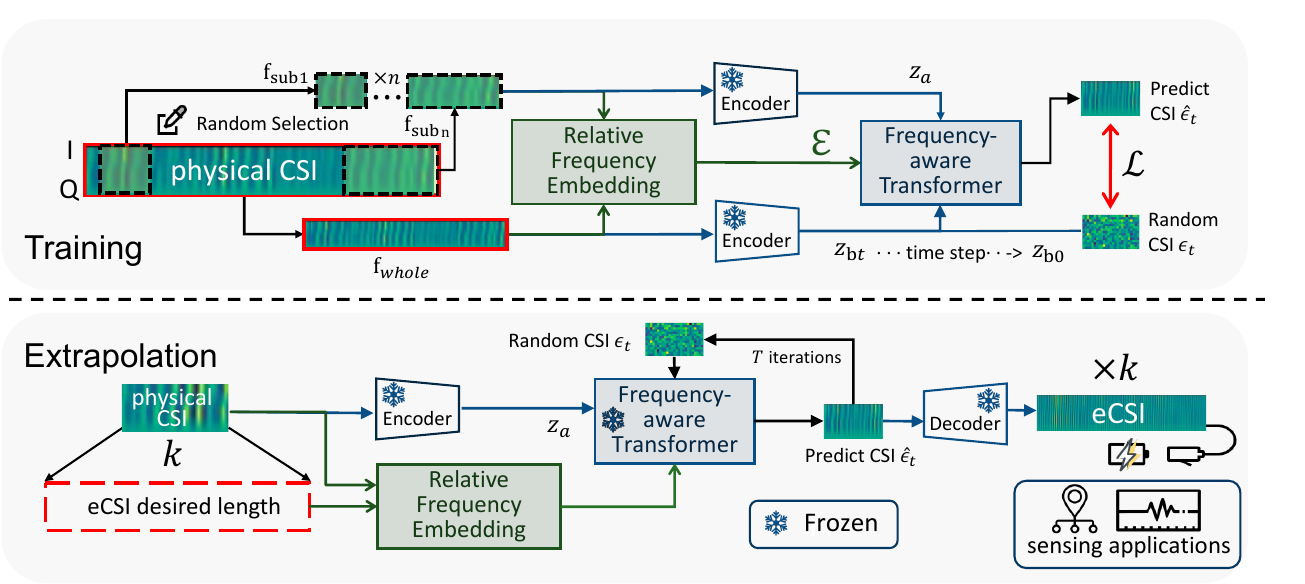}
  %{Figures/motiv1.png}
  \caption{Architecture of \networkname.}
%\vspace{-3mm}
  \label{f:network}
\end{figure*}

\vspace{-0.1in}
\section{\sysname Design}
\label{sec:design}

Although we have formulated the NWB into a novel self-conditioned learning problem, developing a solution with deep neural networks remains challenging; a specialized and innovative framework is necessary. Specifically, 
$\mathcal{G}$ must satisfy the following conditions:

(1) It must be able to accept inputs of any length and produce outputs of any length as specified by the configuration \( k \). As a contrast, existing deep neural frameworks can accept inputs of any length but produce outputs of fixed length according to the input.

(2) The network needs to maintain the unique multipath profiles of each sample and reflect these in the \term{eCSI}, whereas existing deep neural frameworks aim to model an optimal multipath distribution that capable of representing a large number of samples.

\fig~\ref{f:network} illustrates the proposed novel deep neural framework, \networkname, which consists of four major components: 

(1) Relative Frequency Embedding: This method calculates the relationships between subcarrier segments of varying lengths through a relative frequency encoding scheme. The embedding is used as a supplement to the input, effectively addressing the issue of varying sample lengths and enabling the extrapolation of CSI to any frequency position during extrapolation. 

(2) Self-conditioned Diffusion: The reverse diffusion process uses the already generated parts of the data as conditions to predict the next most probable data, thus precisely controlling the generation process, which can be considered as a self-conditioned process. \networkname thus innovatively manipulates the reverse diffusion process, allowing the measured CSI to guide the process and deliver sample-specific knowledge to the extrapolated \term{eCSI}.

(3) Frequency-aware Transformer Model: This model uses a transformer as the backbone and incorporates a frequency-aware attention mechanism. This mechanism integrates inputs of various lengths with RF embeddings to effectively model the positional relationships at different scales within the measured CSI. In this context, 'frequency' refers to the carrier frequency of the CSI, not the frequency characteristic of the batch of sample data.

(4) One-step CSI Extroplation: 
At the extrapolation stage, the frequency-aware transformer is frozen. \networkname generates a virtual container with
arbitrary expansion multiples $k$ based on the relative frequency embedding.

This section details the four key components of \sysname.  

\subsection{Relative Frequency Embedding}
%missing subcarrier
%
\begin{figure}[t!]
  \centering
  \includegraphics[width=0.98\linewidth]{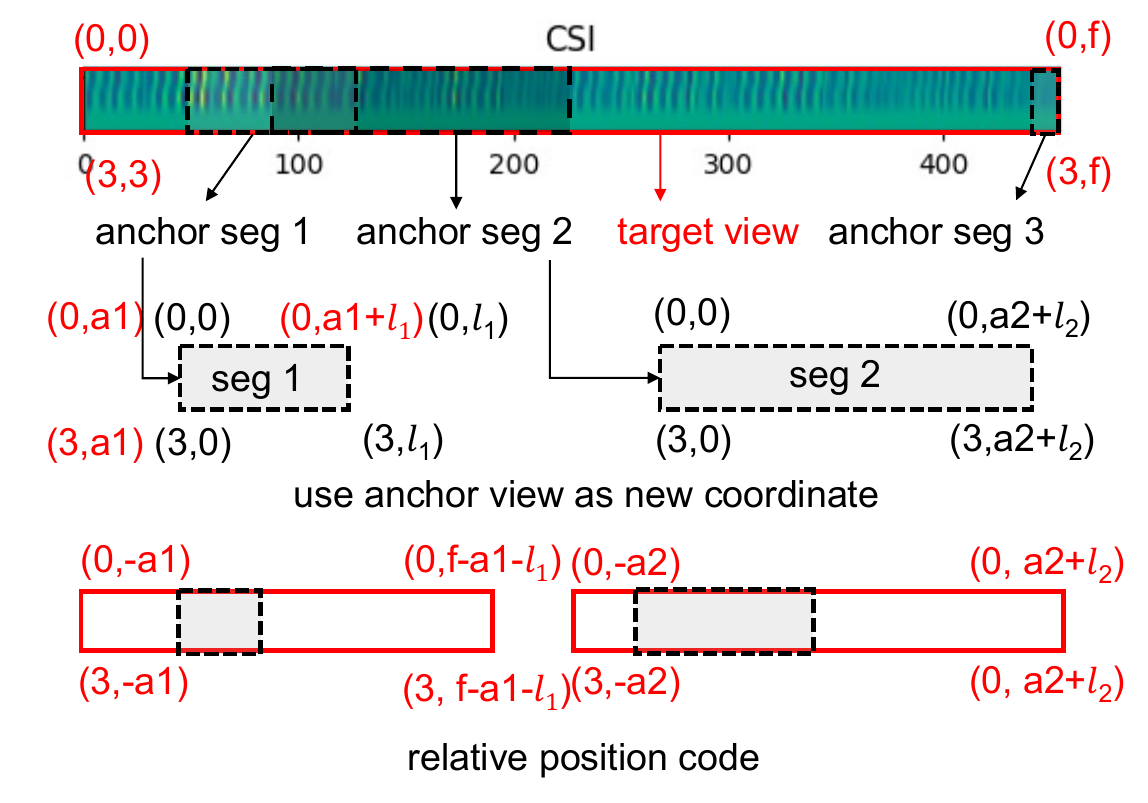}
  %{Figures/motiv1.png}
  %\vspace{-0.9\baselineskip}
  \caption{Detailed mechanism of relative frequency embedding (RFE) of \sysname.}
%\vspace{-6mm}
  \label{f:rfe}
\end{figure}

One of the challenges in modeling the relationship between different frequencies is that WiFi has different subcarrier spacing depending on the protocol and the configured channel bandwidth. To tackle the challenge and to make the learning procedure more elastic, \networkname designs a Relative Frequency Embedding (RFE) module. The RFE aims to represent any positional and frequency relationship between the smaller sub-band CSI and the wider band CSI so that in the extrapolation stage, the \networkname can extrapolate CSI in controllable and continuous multiples $k$ in one step.  

Particularly, each input CSI is organized as $\boldsymbol{H} \in \mathbb{R}^{f\times 3}$, where $f$ is the length of CSI and each CSI includes the real and imaginary parts of the CSI along with the ground truth subcarrier frequency value stacked as a dimension.
During training, for each CSI sample, the RFE module initiates by randomly selecting $n$ subsegment of $\boldsymbol{H}$, i.e., $\boldsymbol{\Tilde{H}} \in \mathbb{R}^{\Tilde{f} \times 3}$. The length ,$\Tilde{f}$, of each $\boldsymbol{\Tilde{H}}$ from the subsegment is randomly chosen to be between $[0.05,0.50]$ times the length of $f$. 

As we aim for \networkname to eventually accept physically measured CSI of arbitrary length as input, we have strategically established coordinates based on the randomly selected subsegment $\boldsymbol{\Tilde{H}}$, designated as the anchor view. We then calculate the relative coordinates of the entire physical CSI,  $\boldsymbol{H}$, designated as the observed view.
This means that for each randomly selected subsegment, a corresponding coordinate is marked to represent the position of the physical CSI from that specific perspective - from a short subsegment to compute a longer entire physical CSI. Given the explicit subcarrier's frequency information, this method enables the calculation of the relative frequency and positional information for each subsegment and the entire physical CSI. The whole process is encapsulated as a relative frequency embedding module that computes the prior information of the relative frequency embedding. 

As illustrated by the Figure~\ref{f:rfe}, for each randomly selected subset segmented CSI, \networkname notes it as an anchor with coordinate $\boldsymbol{p}_A = {p_{Ai},p_{Aj},p_{Ah},p_{Aw}}$, where ${i,j,h,w}$ denotes top, left, height and width position of each CSI matrix ($\boldsymbol{H}$ and $\boldsymbol{\Tilde{H}}$).
The coordinate of original CSI $\boldsymbol{p}_O = {p_{Oi},p_{Oj},p_{Oh},p_{Ow}}$ is calculated based on each anchor coordinate.  For example, the seg1 in Figure~\ref{f:rfe} is $\boldsymbol{p}_A ={0,0,3,l_1}$ whereas the coordinate of original CSI is $\boldsymbol{p}_O = {0,-a1,3,f-a1-l_1}$ To further prepare the relative coordinate for network training purposes, we encode the relative coordination following the cos-sin absolute position embedding in MAE~\cite{he2022masked}.

\subsection{Self-conditioned Diffusion}
\label{s:diff}

\networkname innovatively manipulates the diffusion model, using shorter CSI segments as a self-conditioning mechanism, enabling it to translate current multipath profile - $\Omega$ into the extrapolated \term{eCSI}. This design ensures that each \term{eCSI} shares the same physical correspondence with its respective input physical CSI. 

Similar to the conventional reverse diffusion process, which progressively generates meaningful image content from noise, where the intermediate image in fact acts as the self-conditioning agent~\cite{rombach2022high}, \networkname takes inspiration from this approach. \networkname uses a batch of CSI subsegments of varying lengths, each serving as a self-condition, to generate the entire CSI.
In particular, given the sub-segment CSIs, the entire physical CSI and their output of the RFE, \networkname first encodes each sub-segment and entire CSI by a frozen encoder. Note that the encoder can be any backbone and \networkname does not train it, it just serves as a compressor with a transformation mapping to transfer the sub-segment CSI and the entire CSI to the compact latent space, where they inherently share the same multipath profile, resulting in $z_a$ and $z_b$. $z_a$ is the random sub-segment CSI and $z_b$ denotes the entire measured CSI. The compression aims to improve the efficiency and speed up the convergence of diffusion models~\cite{rombach2022high}.
The forward process of \networkname happens on the entire CSI $z_b$ such that:
\begin{align}
\label{eqn:forward}
&q\left(\mathbf{z}_{b_t} \mid \mathbf{z}_{b_{t-1}}\right)=\mathcal{G}\left(\mathbf{z}_{b_{t-1}} \mid \sqrt{\alpha_t} \mathbf{z}_{b_{t-1}}, \beta_t \mathbf{I}\right), \text { and } \nonumber \\
&q\left(\mathbf{z}_{b_t} \mid \mathbf{z}_{b_0}\right)=\mathcal{G}\left(\mathbf{z}_{b_t} \mid \sqrt{\bar{\alpha}_t},\left(1-\bar{\alpha}_t\right) \mathbf{I}\right),
\end{align}
where $z_{b_0}$ is the original $z_b$, $\bar{\alpha}_t=\prod_{t=1}^t \alpha_i$ and $t$ is our hyperparameter to control the number of steps. $q\left(\cdot \mid \cdot \right)$ is a representation of a Markov chain, $\mathcal{G}(0,1)$ denots standard Guassion noise and $\alpha$ and $\beta$ represent the noise schedule, $\alpha + \beta = 1$.

%之前patch一下是因为不同的随机选择的sample.现在要解patch
For the backward process, given a trained neural network $\mathcal{G}_\theta$ (to be detailed later) takes as inputs subsegment CSI, $z_a$ , noisy entire CSI $z_{b_t}$,  and relative frequency embedding $\mathcal{E}$. Then, the diffusion aims to predict the added noise $\epsilon_t$ on the latent physical CSI $H$ with extrapolated length. The loss function for the model can be written as:
\begin{equation}
\label{eqn:loss}
\mathcal{L}=\left\|\widetilde{\epsilon}_t-\epsilon_t\right\|_2^2 \text {, and } \widetilde{\epsilon}_t=\mathcal{G}_\theta \left(\mathbf{z}_a, \mathbf{z}_{b_t}, \mathcal{E}, t\right) \text {. }
\end{equation}
\networkname solves the diffusion by score-based model and works by learning to transform a standard normal distribution into an empirical data distribution through a sequence of refinement steps, resembling Langevin dynamics~\cite{song2021scorebased}. Consequently, during the extrapolation stage, the \networkname is able to draw the extrapolated \term{eCSI} condition on current input CSI via this reverse diffusion process. 
%这里还可以再展开写，先按下不表
\subsection{Frequency-aware Transformer Model}
Here we detail the network $\mathcal{G}_\theta$ used to incorporate RFE, batches of sub-segment CSI and the diffusion process. 
In cooperation with RFE, \networkname proposes a frequency-aware cross-attention mechanism between RFE and two different sizes CSI, which helps \networkname to execute extrapolation in only one step for any multiple settings. In particular, we first concatenate the noisy entire CSI $z_{b_t}$, sub-segment CSI sequence $z_a$ at the channel dimension, followed by a linear layer to reduce the dimension to reduce the computation cost, and we denote as $z_{g}$, which later is fed into the transformer block~\cite{vaswani2017attention}. The frequency-aware cross-attention mechanism is proposed to learn the different scales of positional relationship such that: 
%CNN,RNN不行
\begin{equation}
\mathbf{z}_d=\operatorname{Attn}\left(\mathbf{Q}_{\mathbf{E}}, \mathbf{K}_{\mathbf{z}_g}, \mathbf{V}_{\mathbf{z}_g}\right)=\operatorname{Softmax}\left(\frac{\mathbf{Q}_{\mathbf{E}} \mathbf{K}_{\mathbf{z}_{\mathbf{g}}}^{\top}}{\sqrt{D}}\right) \mathbf{V}_{\mathbf{z}_g},
\end{equation}
where $\mathbf{Q}_{\mathbf{E}} = \mathcal{E}\mathbf{W_{q}}$, 
$\mathbf{K}_{\mathbf{z}_g} = \mathbf{z}_g\mathbf{W_k}$, 
$\mathbf{V}_{\mathbf{z}_g} = \mathbf{z}_g\mathbf{W_{v}}$ and $\mathbf{W_{q}},\mathbf{W_{k}}, \mathbf{W_{v}}$ are all learnable parameters. After that, the $\mathbf{z}_d$ is directly fed into the transformer decoder, followed by a convolutional layer to predict noise. 

% \begin{algorithm}
% \label{alg:training}
%  Training:

%  \begin{algorithmic}[1]
%  \Require{ measured physical CSI $\boldsymbol{H}$, the frequency-aware transformer for training: $\mathcal{G}_{\theta}$, encoder and decoder $F_{E}, F_{D}$, timestep $T$.}
%  \Ensure{The pretrained network $\mathcal{G}_{\theta}$.}
%  Initialize $\mathcal{G}_{\theta}$.
 
% \While{Converge}{
%     Randomly sample $t$ from $1\sim T$.
    
%     $\boldsymbol{\Tilde{H}} \leftarrow$ Randomselect($\mathbf{H}$), and \boldsymbol{H}
    
%     Compute the relative position embedding $\mathbf{E}$.
    
%     Get the compact $\mathbf{z}_{a}, \mathbf{z}_{b} \leftarrow F_{E}(\boldsymbol{\Tilde{H}},\boldsymbol{H})$.
    
%     Randomly sample Gaussian noise $\epsilon_{t}$.
    
%     Add noise $\epsilon_{t}$ to obtain $\mathbf{z}_{b_{t}}$ via Eq.\ref{eqn:forward}.
    
%     Predict $\widetilde{\epsilon}_t$ by $\mathcal{G}_{\theta}(\mathbf{z}_{a}, \mathbf{z}_{b_{t}}, t)$.
    
%     Compute the loss $\Vert \widetilde{\epsilon}_t - \epsilon_t \Vert_p^p$ via Eq.\ref{eqn:loss}. 

%     Backward and update the network $\mathcal{G}_{\theta}$.
%   }
% \KwOut{$\mathcal{G}_{\theta}$}
% % \KwOut{\ensuremath{\mathcal{G}_{\theta}}}
% Extrapolation:

% \KwIn{multiple $k$, measured physical CSI $\boldsymbol{H}$}
% \KwResult{The extrapolated \term{eCSI}}
% Get the relative position embedding $\mathbf{E}$ by computing $\boldsymbol{H}$ of \term{eCSI} through $k$.

% Randomly sample Gaussian noise with the size of \term{eCSI} 

% 519
% {T = 0} {Denoise via Eq.\ref{eqn:reverse}} 

% \caption{\sysname \networkname}
% \end{algorithmic}
% \end{algorithm}
\begin{algorithm}[t]
\caption{\sysname\ \networkname}
\label{alg:training}

% ---------- Training ----------
\textbf{Training:}\;

\KwIn{
measured physical CSI $\boldsymbol{H}$; 
frequency-aware transformer $\mathcal{G}_{\theta}$;
encoder/decoder $F_E, F_D$; timestep $T$
}
\KwOut{pretrained network $\mathcal{G}_{\theta}$}

Initialize $\mathcal{G}_{\theta}$\;

\While{not converged}{
  Randomly sample $t \in \{1,\dots,T\}$\;

  $\widetilde{\boldsymbol{H}} \leftarrow \text{Randomselect}(\boldsymbol{H})$\;
  Compute the relative position embedding $\mathbf{E}$\;

  Get compact features $\mathbf{z}_a, \mathbf{z}_b \leftarrow F_E(\widetilde{\boldsymbol{H}}, \boldsymbol{H})$\;

  Randomly sample Gaussian noise $\epsilon_t$\;

  Add noise to obtain $\mathbf{z}_{b_t}$ via Eq.~\ref{eqn:forward}\;

  Predict $\widetilde{\epsilon}_t \leftarrow \mathcal{G}_{\theta}(\mathbf{z}_a, \mathbf{z}_{b_t}, t)$\;

  Compute loss $\lVert \widetilde{\epsilon}_t - \epsilon_t \rVert_p^p$ via Eq.~\ref{eqn:loss}\;

  Backward and update $\mathcal{G}_{\theta}$\;
}

\BlankLine

\KwOut{$\mathcal{G}_{\theta}$}
Extrapolation:

\KwIn{multiple $k$; measured physical CSI $\boldsymbol{H}$}
\KwResult{the extrapolated \term{eCSI}}

Get the relative position embedding $\mathbf{E}$ for \term{eCSI} using $k$ and $\boldsymbol{H}$\;

Randomly sample Gaussian noise with the size of \term{eCSI}\;

Set $T \leftarrow 0$ and denoise via Eq.~\ref{eqn:reverse}\;
\end{algorithm}

\subsection{One-step CSI Extroplation}
Once \networkname is trained, one can extrapolate CSI in any controlled multiples $k$, since the designed RPE can represent any positional relationship between the physical CSI and the extrapolated \term{eCSI}. In particular, the input physical CSI is treated as the anchor and together with $k$, \networkname generates a virtual container, $\widetilde{\mathbf{z}}_{b_0}$, which is filled with random noise. Then \networkname predicts the noise, compute the  $\widetilde{\mathbf{z}}_{b_0}$, and predict $z_{b_{t-1}}$ step-by-step as mentioned in \sec~\ref{s:diff}, which can be formulated as:
\begin{equation}
\begin{aligned}
& q\left(\mathbf{z}_{b_{t-1}} \mid \mathbf{z}_{b_t}, \widetilde{\mathbf{z}}_{b_0}\right)=\mathcal{G}\left(\mathbf{z}_{b_{t-1}} ; \widetilde{\mu}_t\left(\mathbf{z}_{b_t}, \widetilde{\mathbf{z}}_{b_0}\right), \widetilde{\beta}_t \mathbf{I}\right), \\
& \widetilde{\mu}_t\left(\mathbf{z}_{b_t}, \widetilde{\mathbf{z}}_{b_0}\right)=\frac{\sqrt{\bar{\alpha}_{t-1}} \beta_t}{1-\bar{\alpha}_t} \widetilde{\mathbf{z}}_{b_0}+\frac{\sqrt{\alpha}\left(1-\bar{\alpha}_{t-1}\right)}{1-\bar{\alpha}_t} \mathbf{z}_{\mathbf{b}_{\mathbf{t}}}, \\
& \text {and } \widetilde{\beta}_t=\frac{1-\bar{\alpha}_{t-1}}{1-\bar{\alpha}_t} \beta_t \text {. } 
\end{aligned}
\label{eqn:reverse}
\end{equation}
After several iterations, when $t = 0$, \networkname can obtain the extrapolated \term{eCSI}.

\subsection{One-step CSI Extroplation}
Once \networkname is trained, one can extrapolate CSI in any controlled multiples $k$, since the designed RPE can represent any positional relationship between the physical CSI and the extrapolated \term{eCSI}. In particular, the input physical CSI is treated as the anchor and together with $k$, \networkname generates a virtual container, $\widetilde{\mathbf{z}}_{b_0}$, which is filled with random noise. Then \networkname predicts the noise, compute the  $\widetilde{\mathbf{z}}_{b_0}$, and predict $z_{b_{t-1}}$ step-by-step as mentioned in \sec~\ref{s:diff}, which can be formulated as:
\begin{equation}
\begin{aligned}
& q\left(\mathbf{z}_{b_{t-1}} \mid \mathbf{z}_{b_t}, \widetilde{\mathbf{z}}_{b_0}\right)=\mathcal{G}\left(\mathbf{z}_{b_{t-1}} ; \widetilde{\mu}_t\left(\mathbf{z}_{b_t}, \widetilde{\mathbf{z}}_{b_0}\right), \widetilde{\beta}_t \mathbf{I}\right), \\
& \widetilde{\mu}_t\left(\mathbf{z}_{b_t}, \widetilde{\mathbf{z}}_{b_0}\right)=\frac{\sqrt{\bar{\alpha}_{t-1}} \beta_t}{1-\bar{\alpha}_t} \widetilde{\mathbf{z}}_{b_0}+\frac{\sqrt{\alpha}\left(1-\bar{\alpha}_{t-1}\right)}{1-\bar{\alpha}_t} \mathbf{z}_{\mathbf{b}_{\mathbf{t}}}, \\
& \text {and } \widetilde{\beta}_t=\frac{1-\bar{\alpha}_{t-1}}{1-\bar{\alpha}_t} \beta_t \text {. } 
\end{aligned}
\label{eqn:reverse}
\end{equation}
After several iterations, when $t = 0$, \networkname can obtain the extrapolated \term{eCSI}.

\section{Experiments and Evaluation}
\networkname stands apart from the classic deep learning methods that are based on maximizing likelihood to estimate a set of parameters of an assumed probability distribution, given some observed data. Instead, \networkname relies on the universal approximation theorem and uses sample-level observed data as self-conditioning to adapt and translate its physical information into the extrapolated eCSI. As a result, our approach should not rely on a particular dataset or channel conditions. To validate the efficacy, we evaluate the system with different WiFi protocols, different hardware devices (with different NIC cards), using both self-collected and publicly available datasets as well as WiFi CSI and UWB CSI. The detailed information is listed in Table~\ref{t:settings}.

\begin{table}[h]
\footnotesize
\centering
\caption{Summary of evaluation settings.}
\label{t:settings}
\begin{tabular}{c|c|c|c}
\hline
Device       & Protocol       & Bandwidth      & Dataset           \\ \hline
ath9k-AR9344 & 802.11n 2.4GHz & 40MHz & HandFi-gesture~\cite{10.1145/3625687.3625812}     \\
RPI-AX200    & 802.11n 5GHz   & 20MHz          & NWB-localization  \\
Pico-AX210   & 802.11ax 6GHz  & 160MHz         & NWB-vital sign    \\
Quantenna AP & 802.11ac 5GHz  & 80MHz          & DLoc-localization~\cite{ayyalasomayajula2020deep}           
\end{tabular}
\vspace{-6mm}
\end{table}

\subsection{Evaluation Setup}

\begin{figure*}[t!]
       \subfloat[RasPi-AX200]{
    \label{subfig:ax200}
  \includegraphics[width=0.15\linewidth]{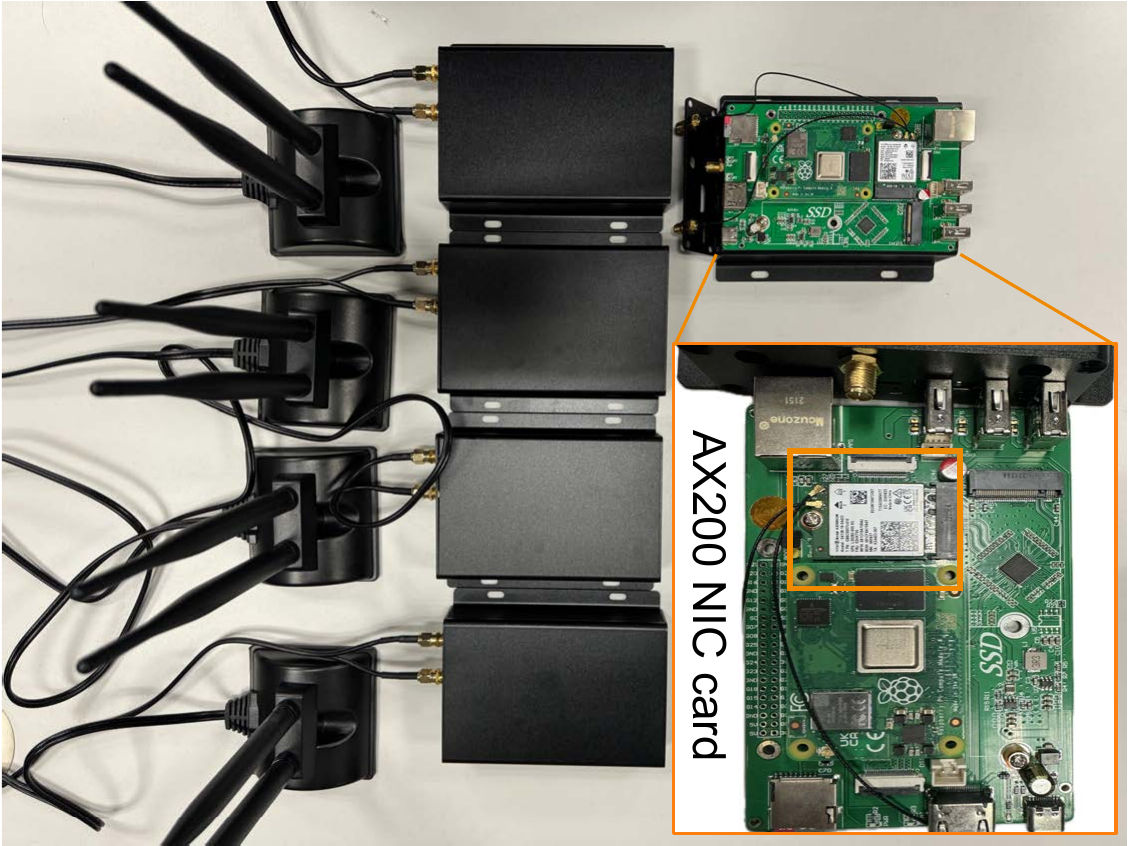}}
  \hfill
    \subfloat[Picoscences-AX210]{
    \label{subfig:210}
  \includegraphics[width=0.15\linewidth]{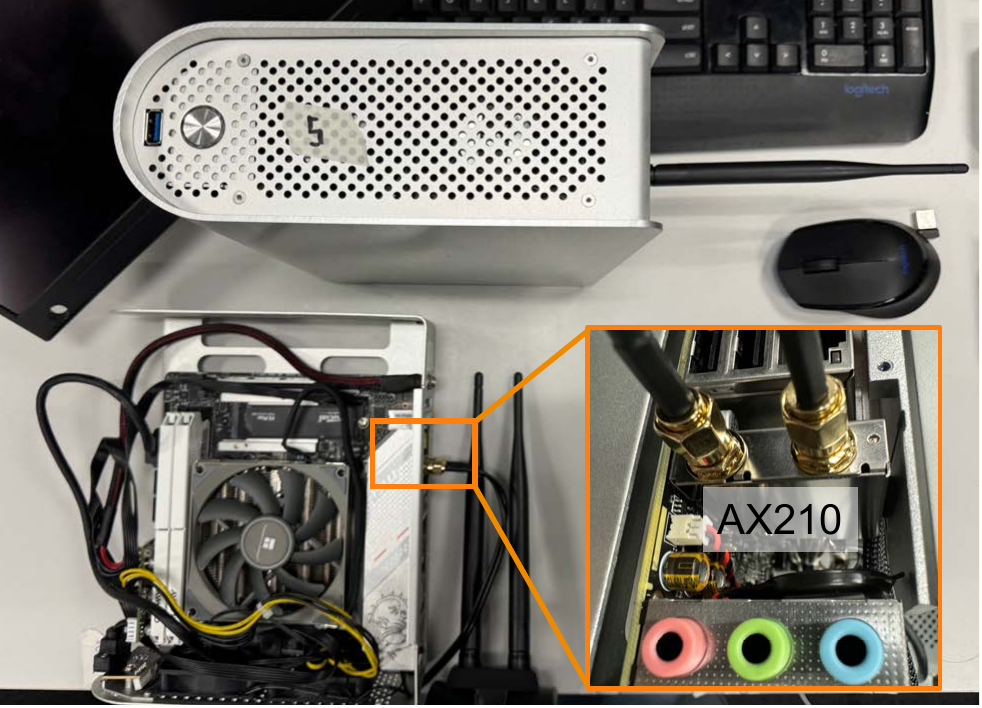}}
  \hfill
   % \vspace{-0.9\baselineskip}
%     \caption{\sysname Data Collection Platforms}\label{f:hardware}
% \end{figure}
% \begin{figure*}[t!]
    \subfloat[Hallway]{
    \label{subfig:hallway}
    \includegraphics[width=0.15\linewidth]{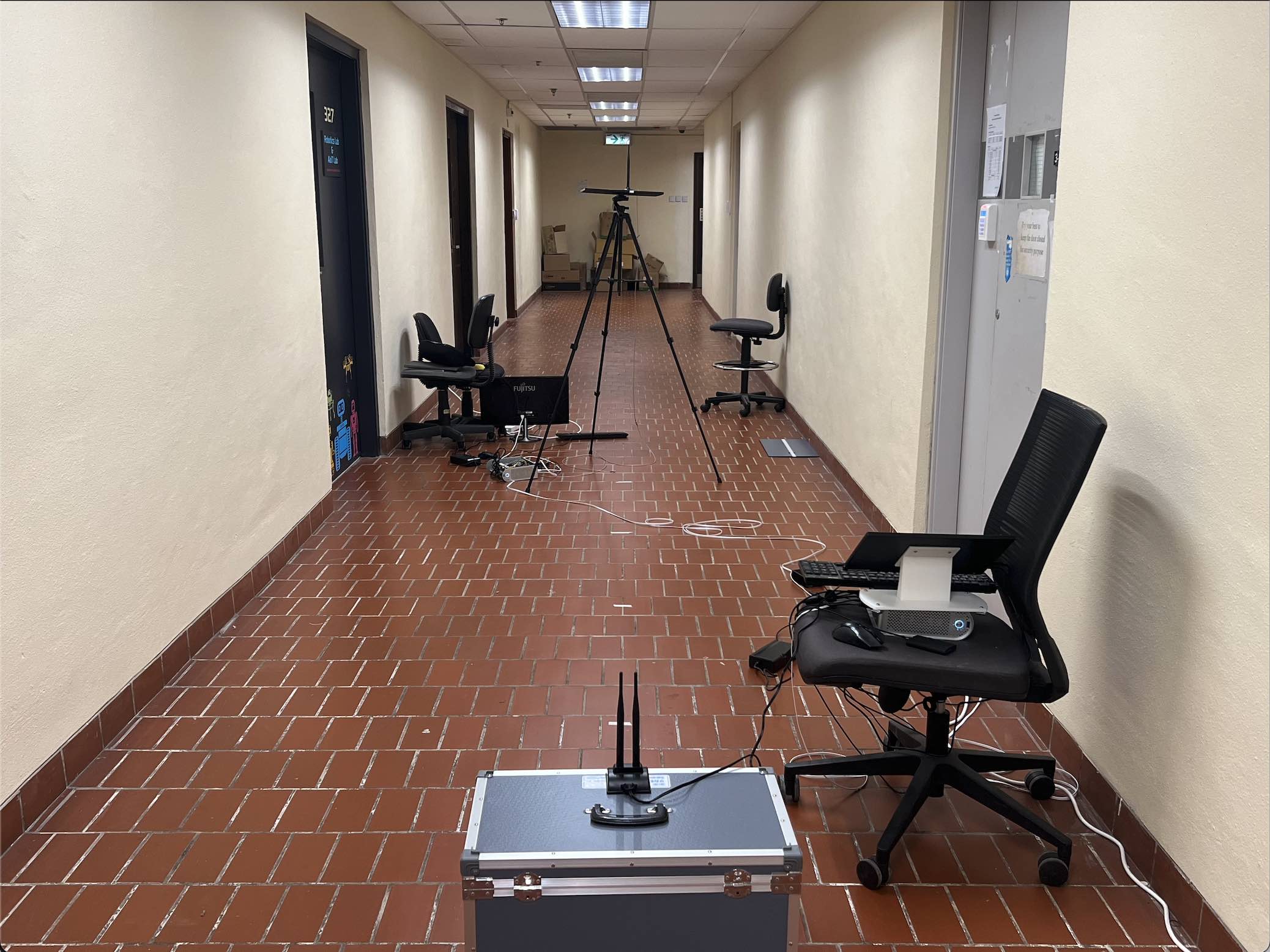}}
   \hfill
    \subfloat[Classroom]{
    \label{subfig:classroom}
  \includegraphics[width=0.15\linewidth]{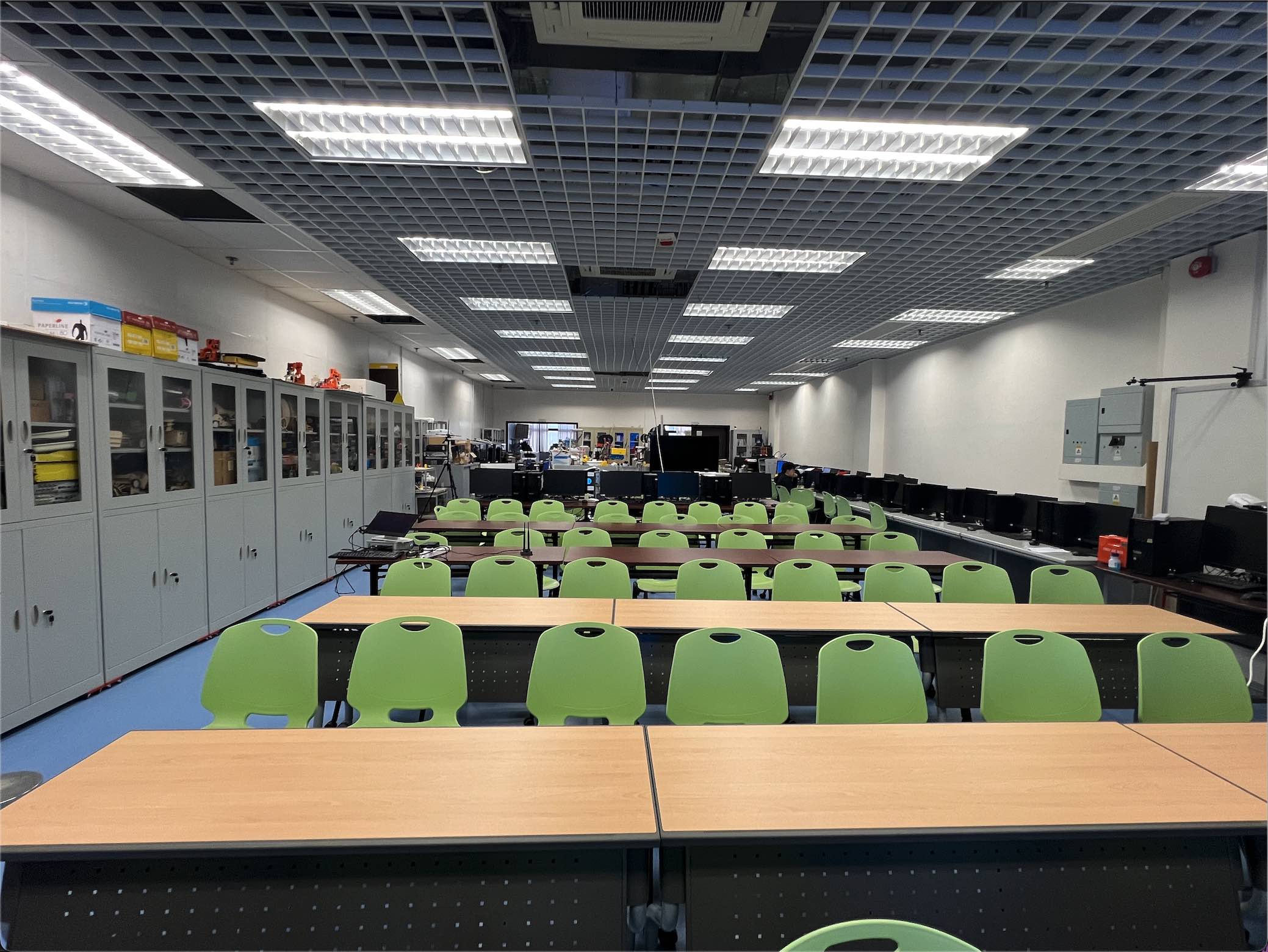}}
   \hfill
       \subfloat[Office]{
    \label{subfig:office}
  \includegraphics[width=0.15\linewidth]{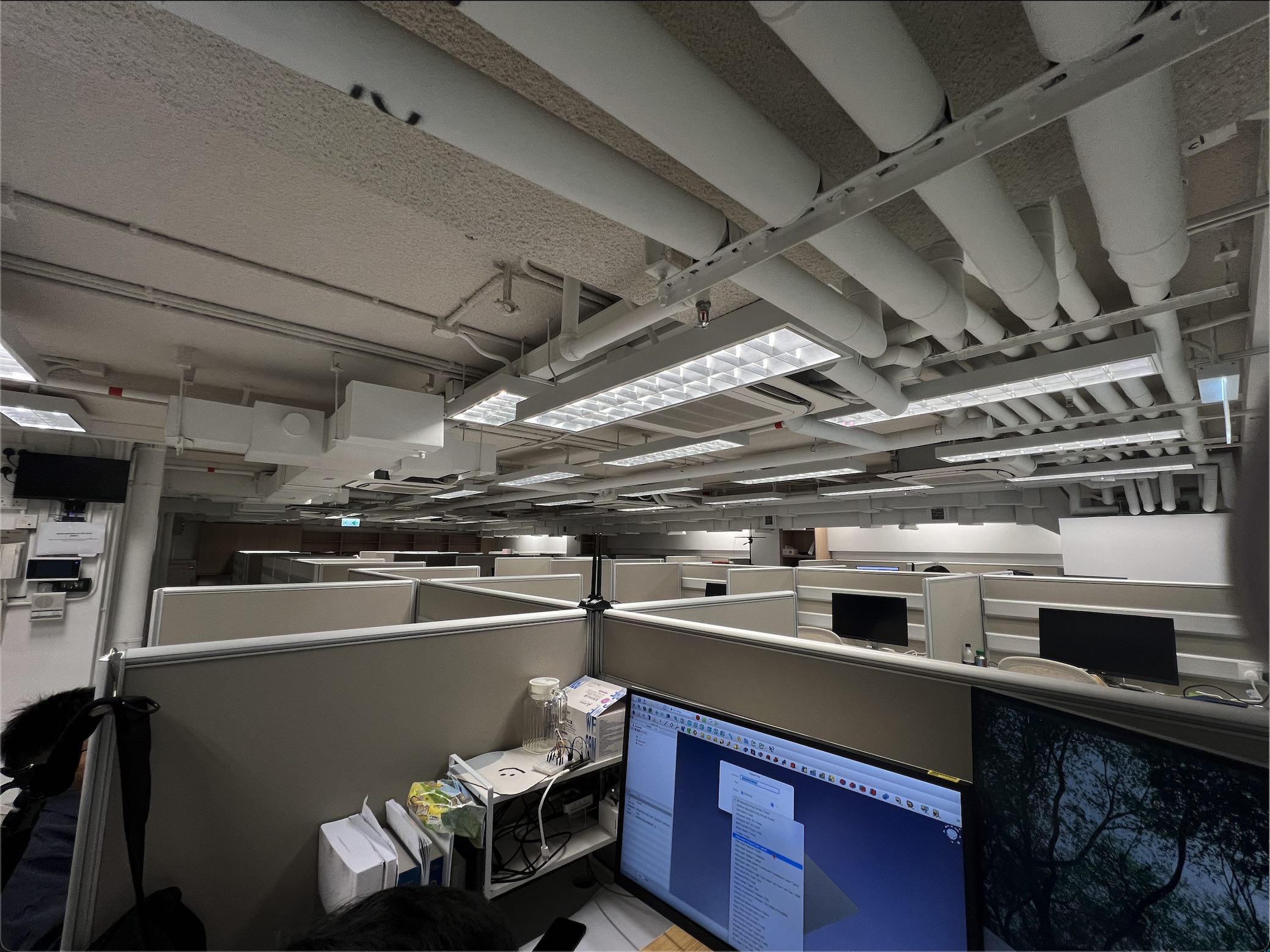}}
  \hfill
    \subfloat[Meeting Room]{
    \label{subfig:meeting}
  \includegraphics[width=0.15\linewidth]{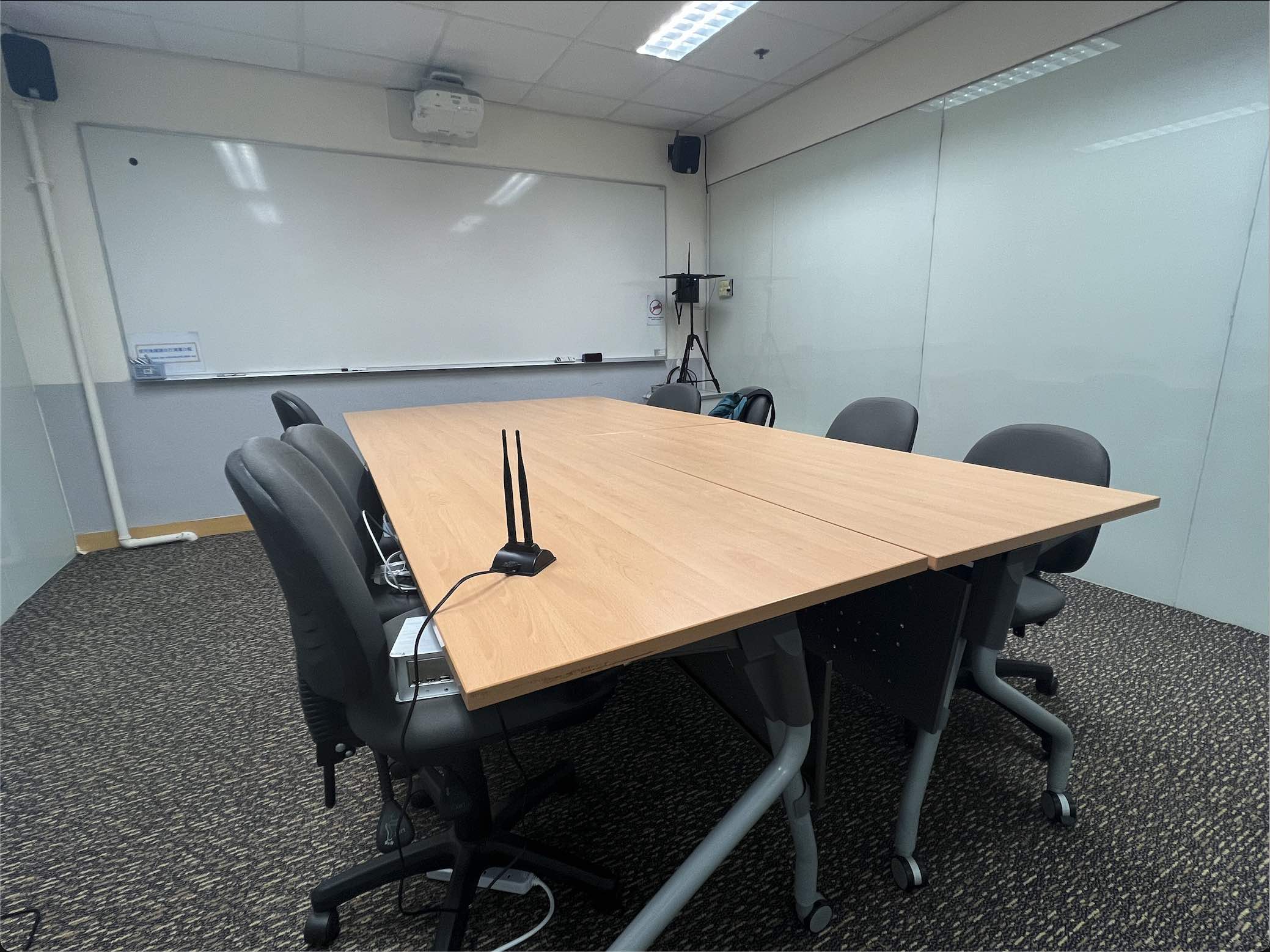}}
  %  \vspace{-0.9\baselineskip}
    \caption{\sysname Data Collection Hardware: (a)(b); Scenario: (c) a natural hallway (d) a highly dynamic classroom (e) a natural office (f) a static meeting room.} \label{f:envs}
   % \vspace{-0.1in}
\end{figure*}
\head{Ground Truth Construction for eCSI Validation}
%\head{Data Collection and Datasets}
Since no publicly available dataset of 160MHz was found, we modified the INTEL IWLMVM driver for AX200 Network Card (NIC) to enable CSI recording on the Raspberry Pi Compute Module 4 (RPI) installed on an extension board (Figure~\ref{subfig:ax200}). The device is set to sub-6GHz bands with 160MHz bandwidth. To further access the whole 1GHz band on sub-7GHz, we used two NUC Mini PCs with AX210 NIC driven by PicoScenes~\cite{picos} (Figure~\ref{subfig:210}). The PicoScenes tool was configured in initiator-responder mode, enabling automatic sweeping across seven consecutive 160MHz channels in total. Both platforms are configured as one transmission antenna and two receiver antennas, therefore one transmission will yield two samples. The sampling rate is configured as 100Hz.
During inference, only 20/40/80MHz sub-bands of the 160MHz samples are used as input to Wukong, while the full-band 160MHz CSI serves as the ground truth for evaluation. This enables direct comparison between \term{eCSI} and true wideband CSI.

\head{Dataset for cross-domain generalization}
To evaluate both model fidelity and generalization, we test \sysname not only on self-collected datasets featuring diverse multipath conditions, but also on public datasets collected under different protocols and hardware platforms. Altogether, our data spans a wide range of protocols (802.11n/ac/ax), frequency bands (2.4GHz–6GHz), and hardware setups, including Intel AX200, AX210, Quantenna APs, and DW1000 UWB devices. This experimental setup enables us to assess Wukong’s ability to generalize across heterogeneous physical environments, device types, and signal characteristics—all without retraining.
For public datasets, we strictly follow the original train-test splits as described in their respective papers. For the self-collected NBW dataset, we gather CSI data from four distinct indoor environments (Figure 5): a quiet meeting room (controlled setup), a natural office space, a hallway with intermittent activity, and a highly dynamic classroom during the minutes leading up to a lecture. The latter two in particular reflect non-stationary, challenging multipath conditions with moving users and environmental variability.

In each location, we collect 12,000 CSI samples, allocating 10,000 for training and 2,000 for testing. All CSI data is preprocessed by normalizing amplitude and scaling the frequency dimension to the [0, 1] range (by multiplying the subcarrier frequency in MHz by $10^{-4}$). Each CSI sample is represented as a 3 × 32 × 32 tensor, where the three channels encode the real part, imaginary part, and normalized carrier frequency, respectively. This formatting enables the network to jointly learn amplitude, phase, and frequency structure in a unified representation. Performance is evaluated primarily using Mean Squared Error (MSE) between extrapolated and ground truth CSI.

\begin{figure}[t!]
\centering
  
  \includegraphics[width=0.9\linewidth]{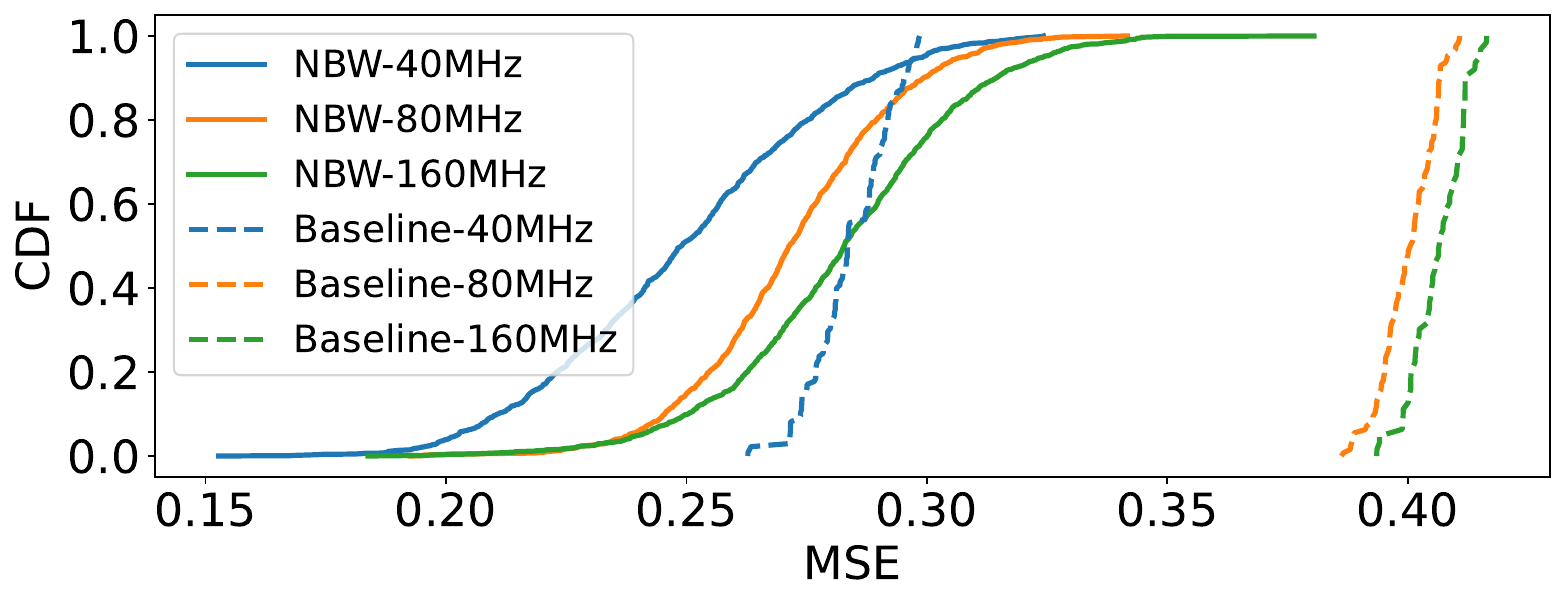}
  \vspace{-0.9\baselineskip}
   % \vspace{-0.9\baselineskip}
    \caption{CDF result of MSE comparison with the baseline of the extrapolated eCSI. Input: 20MHz,  Extrapolated: 40MHz, 80MHz, 160MHz, with ground truth measured.} \label{subfig:baseline}
    %\vspace{-3mm}
\end{figure}

\head{Training Details}
The frozen (non-training) encoder and decoder are fetched from Hugging Face~\cite{Hugging} and disable the normalization layers.  \networkname is implemented with PyTorch on a server equipped with one 4090 GPU. The Frequency Aware Transformer is composed of 4 stacked transformer blocks with a cross-attention block. A convolutional layer is adopted at the end to smooth the extrapolated eCSI. Due to the training strategy of the network, it is designed to be able to extrapolate the entire CSI from each sub-band CSI in each batch, which consists of sub-bands of different random lengths, and generate the full-band CSI from random full-band CSI noise through diffusion. Therefore, except for epoch, we have another parameter that is diffusion timesteps, which is set as 200. In total, we trained 20,000 steps with 100 epochs. The AdamW optimizer~\cite{loshchilov2017decoupled} is adopted to train the network with a learning rate equal to 0.02, weight decay set to 0.01, and betas set to 0.99. A CosineAnnealingLR scheduler~\cite{loshchilov2016sgdr} is also adopted, with [T\_max=epochs, T\_warmup=10, eta\_min=0]. Although training Wukong takes several hours, its inference is highly efficient, requiring only ~0.5 ms per sample on an 4090 GPU. This is significantly faster than typical classification or regression tasks in sensing pipelines, making Wukong well-suited for real-time WiFi sensing applications.

\begin{figure*}[t]
\centering
    \subfloat[$20\,\mathrm{MHz} \rightarrow 40\,\mathrm{MHz}$]{
    \label{subfig:40}
    \includegraphics[width=0.15\linewidth]{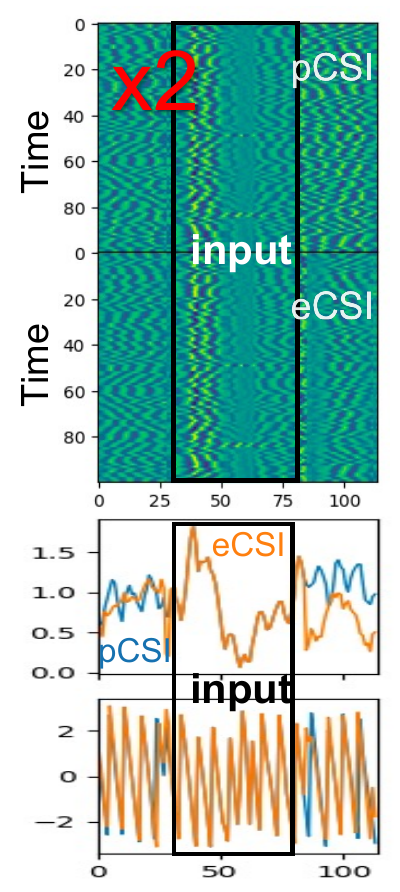}}
   \hspace{0.1em}
    \subfloat[$20\,\mathrm{MHz} \rightarrow 80\,\mathrm{MHz}$]{
    \label{subfig:80}
  \includegraphics[width=0.245\linewidth]{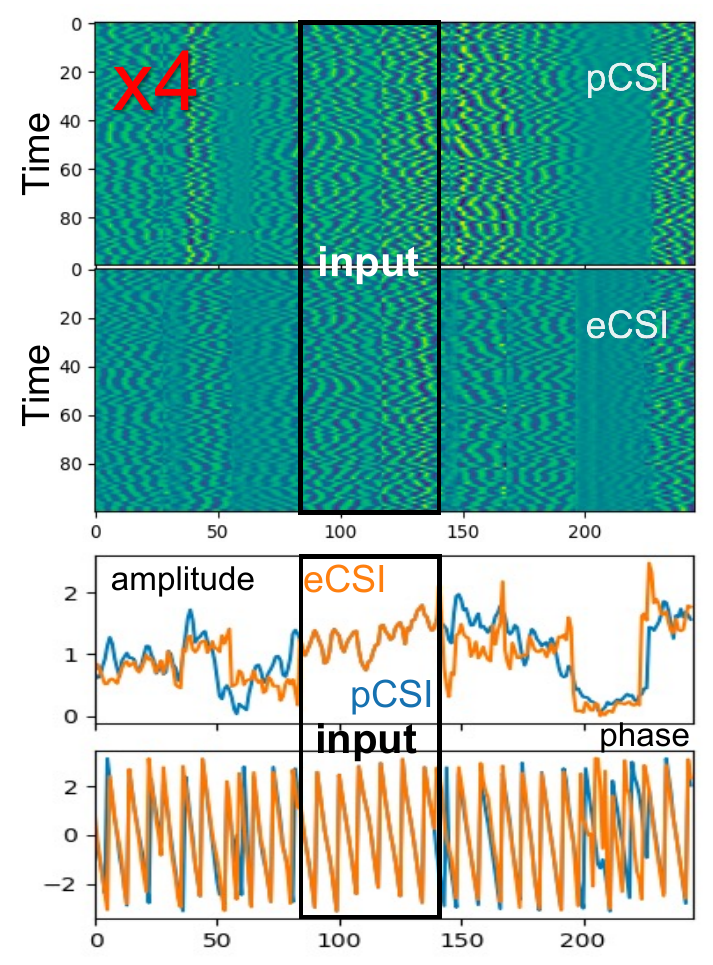}}
  \hspace{0.1em}
       \subfloat[$20\,\mathrm{MHz} \rightarrow 160\,\mathrm{MHz}$]{
    \label{subfig:160}
  \includegraphics[width=0.415\linewidth]{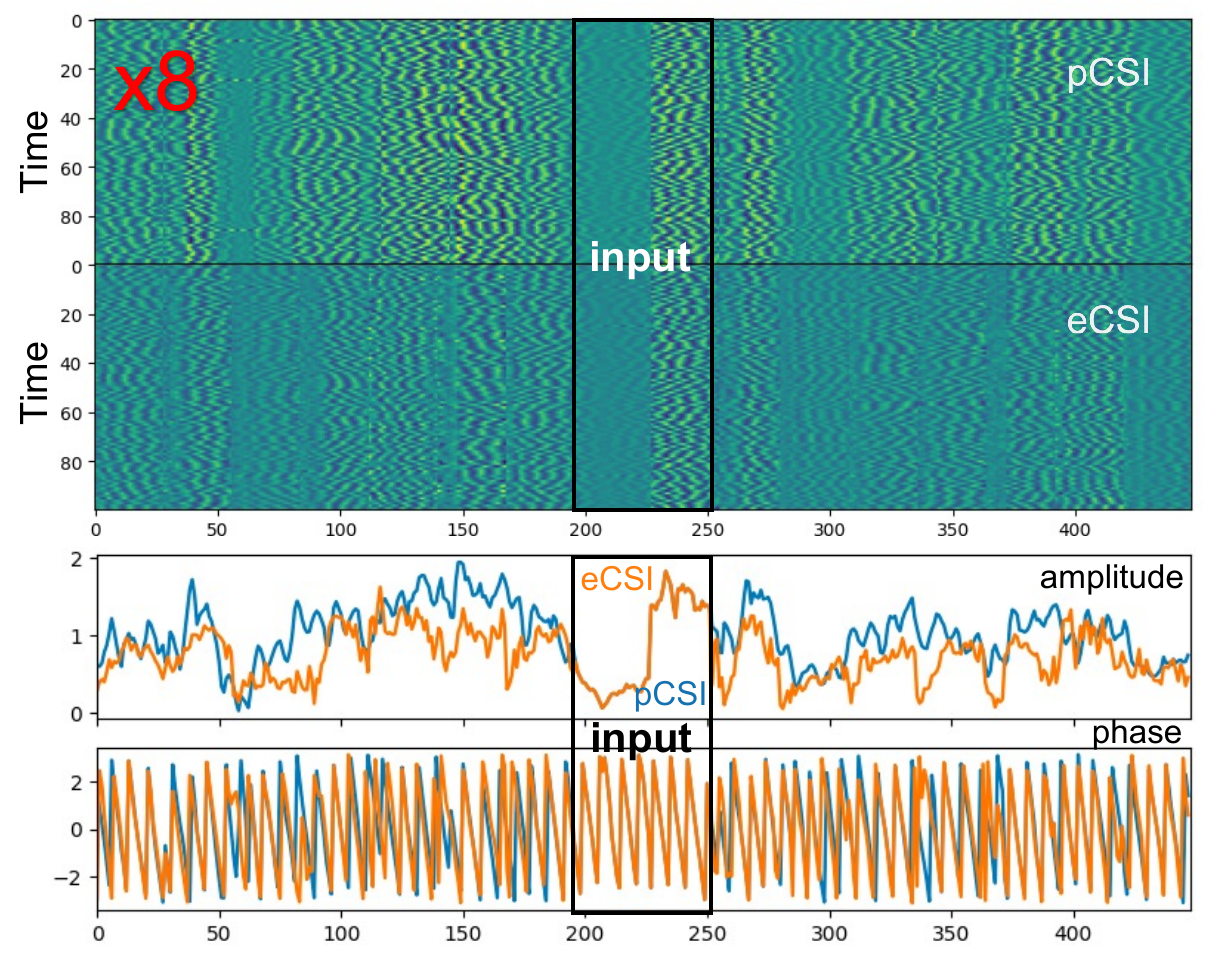}}
  
  % \vspace{-0.9\baselineskip}
    \caption{Visualization of eCSI ($\times 2,\times 4,\times 8$) of data samples collected with 6025MHz central frequency.}\label{f:visuals}
    %\vspace{-3mm}
\end{figure*}

\subsection{Baseline Comparison}
Due to the absence of a method capable of achieving continuous, larger bandwidth through \textbf{a single inference} while using a limited bandwidth as input, we select the state-of-the-art (SOTA) deep generative-based channel mapping method FIRE~\cite{liu_fire_2021} as the baseline. FIRE employs a VAE~\cite{kingma2013auto} to map the downlink channel based on the uplink channel, assuming channel spatial reciprocity. To adapt the FIRE to achieve NBW concept, we divide the entire band equally and use the first sub-bandwidth for the uplink channel, while the remaining portion is utilized for learning the downlink channel. The larger bandwidth is obtained by splicing those downlink channels\footnote{UWB-FI~\cite{li2024uwb} does not output CSI but instead provides higher resolution AoA-ToF tuples, and it requires input from data across different frequency bands through fast channel hopping, which makes it different from our problem so as not able to use as a baseline.}. 
Figure~\ref{subfig:baseline} reports the MSE results obtained from bandwidth multiples of x2, x4, and x8, corresponding to 40MHz, 80MHz, and 160MHz, respectively, using the NBW-hallway dataset with an input bandwidth of 20MHz. We have three important findings. 
Firstly, the average MSE result of x8 eCSI from \networkname is close to the average MSE result of x2 times extrapolation from the baseline method. Secondly, the MSE increases cumulatively with the increase of extrapolation multiples.
Lastly, the baseline method exhibits a substantially larger error in spliced results compared to its standalone predictions, as evidenced by the significant disparity between x4 and x2.

\begin{table}[t!]
		\centering
 \footnotesize
 %\vspace{-0.9\baselineskip}
 \caption{MSE and AccCIR across }
% \vspace{-0.9\baselineskip}
\addtolength{\tabcolsep}{-4pt}
\label{t:overall}
\begin{tabular}{l|cc|c|cccc}
\toprule
Bandwidth & \multicolumn{2}{c|}{x2} & x4 & \multicolumn{4}{c}{x8(NWB)}    \\ \hline
Dataset   & DLoc           & HandFi           & DLoc          & Hallway &  Classroom & Office& Meeting Room \\ \hline
MSE↓       &     0.223           &    0.271              &     0.274          &   0.272      &  0.224      &   0.321        &   0.232           \\ \hline
AccCIR↑    &    0.756            &   0.742               &   0.738            &   0.743      &  0.777      &    0.717       &  0.752            \\ \bottomrule
\end{tabular}

 % \caption{MSE and AccCIR across datasets.\com{Where are "examples of AccCIR visualization"?}}
 \label{t:acccir}
 \vspace{-3mm}
\end{table}

To better understand the performance gain of \networkname, we further observed the statistical patterning of eCSI and the magnitude and phase of eCSI.
Figure~\ref{f:visuals} showcases the \sysname, with the upper column displaying a batch of CSI data for overall pattern observation, and the lower column showing one CSI sample for detailed observation. We plot the amplitude and phase of the CSI. The input physical CSI (pCSI) is 20MHz and the eCSI are 40MHz, 80MHz, and 160MHz respectively. It is safe to say that the network is able to capture the underlying multipath profile because no matter 2 times, 4 times, or 8 times results of eCSI, they are similar to the pCSI. It is worth noting that although the overall MSE error increases with the bandwidth of eCSI, we do not observe the expected lower accuracy of eCSI the farther away the bandwidth is. This also indicates that the network extrapolates to different frequencies based on the self-conditioned multipath profile of each sample, rather than extrapolating the input sample with a fixed learned pattern.

%two CIR example.

\subsection{Multipath Resolution}
% \head{Distinguishable path resolution accuracy} 
Since the system is designed to extrapolate the measured physical CSI in a plug-and-play fashion, rather than directly increasing the accuracy of a specific sensing application in an end-to-end manner, we introduce distinguishable path resolution accuracy as another metric, denoted as AccCIR, which is represented as the correlation between the CIR (Channel Impulse Response) profile of eCSI and the ground truth pCSI’s CIR profile. Due to the fact that the real and imaginary parts are separated and input into the network and the input essentially is CFR, the AccCIR reflects the effectiveness of eCSI. \tab\ref{t:acccir} gives the results across datasets and across bandwidths, where overall a lower MSE can represent a higher AccCIR. To visualize the physical meaning of AccCIR, we draw two traces from the classroom and office, respectively.

%\subsection{\networkname Parameters}
% \begin{figure}[t!]
% \centering
% \includegraphics[width=0.8\linewidth]{fig/uwb-smaller.pdf}
%  %  \vspace{-0.9\baselineskip}
%     \caption{Visualization of 60MHz→960MHz eCSI v.s. ground truth pCSI} \label{subfig:uwb}

% % \vspace{-3mm}
% \end{figure}

% \subsection{Scalability of \networkname}
% The \sysname can be extrapolated to arbitrary bands based on the \networkname design, but since we are only able to access a maximum 160MHz WiFi data, we cannot verify the validity of a larger eCSI. Therefore, we resort to UWB signals to demonstrate the scalability of the system.
% To investigate the impact of larger multiple, we adopt a UWB dataset. The raw CIR sample of UWB is read from DW1000, EVK1000 UWB chips and performed FFT
% on the samples to get the channel frequency response (CFR). As the
% CIR sampling interval is $1ns$ and there are 1016 samples in total, the frequency resolution of the FFT bins is roughly 1MHz. Given the bin resolution and carrier frequency are very much different from the WiFi CSI, we retrain the \networkname and use 60MHz subband to extrapolate the whole band, the overall MSE is 0.161 and the AccCIR is 0.781, which demonstrate the scalability of \networkname.
% The lower MSE values are observed due to the fact that the CFR of UWB is more consistent compared to the CSI of WiFi. Figure~\ref{subfig:uwb} visualize the UWB channel frequency response covers from 6.04GHz to 6.94GHz.

\begin{figure}[t]
\centering
    \subfloat[Sub-bands volume]{
    \label{subfig:samples}
    \includegraphics[width=0.33\linewidth]{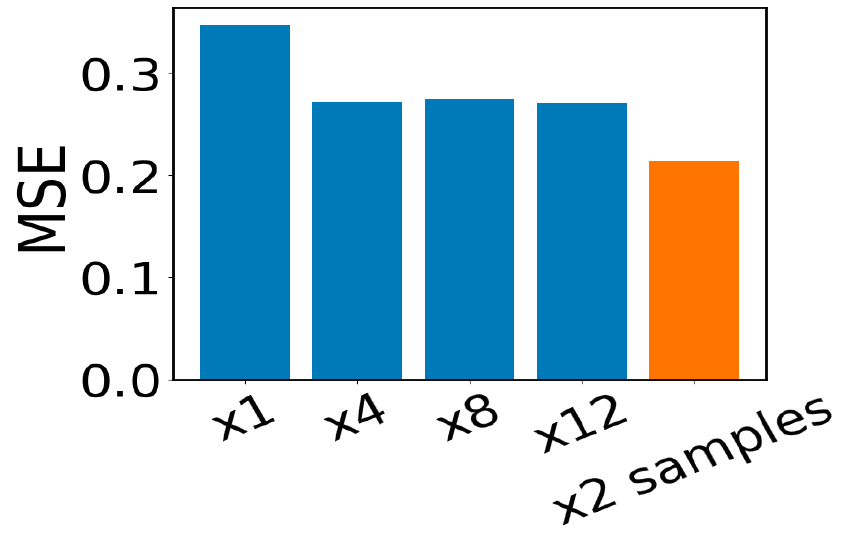}}
   % \hspace{0.1em}
    \subfloat[Time steps]{
    \label{subfig:ts}
  \includegraphics[width=0.33\linewidth]{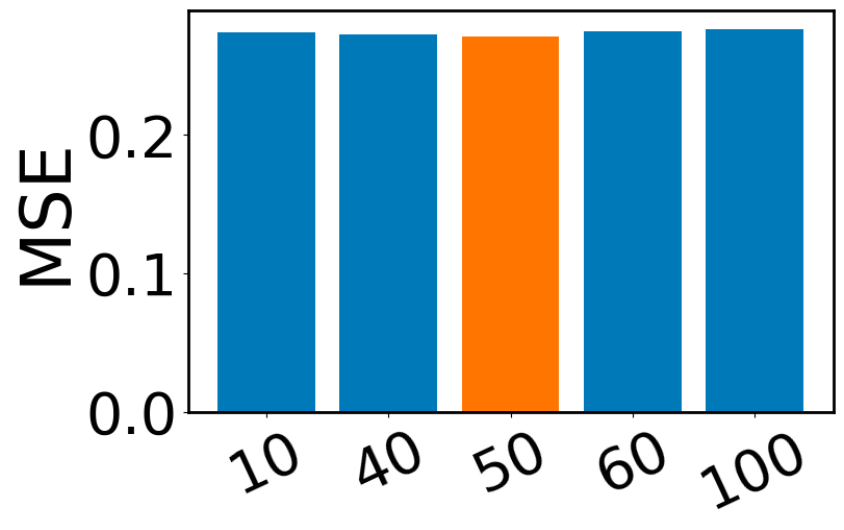}}
  % \hspace{0.1em}
       \subfloat[Batch size]{
    \label{subfig:bs}
  \includegraphics[width=0.33\linewidth]{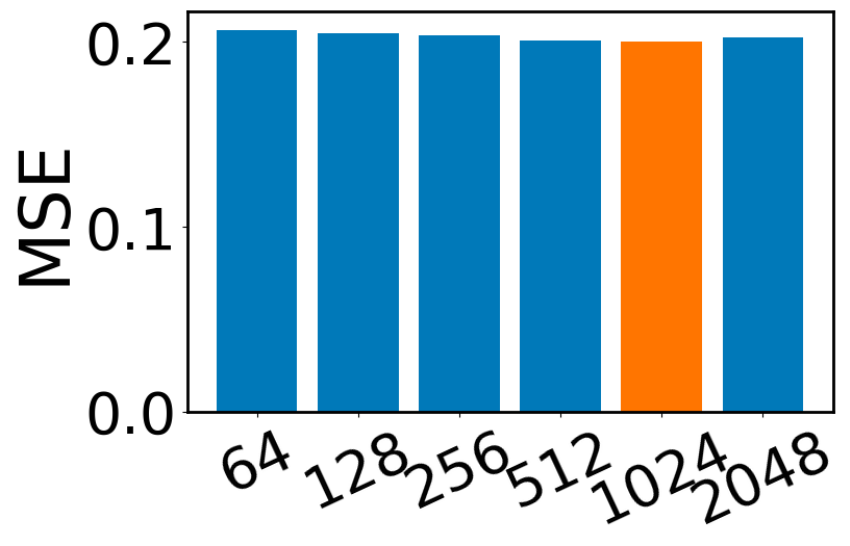}}
  
  % \vspace{-0.9\baselineskip}
    \caption{\networkname parameter search}\label{f:ab}
   % \vspace{-4mm}
\end{figure}

\subsection{Parameter Study}
To understand the impact of \networkname's hyper-parameters, we utilize the data collected in a hallway environment for a grid search experiment. We calculate the mean square error (MSE) between 160MHz eCSI and the ground truth 160MHz physical CSI.

Since the network randomly selects sub-bands for each input sample, we investigate the results of randomly selecting 1, 4, 8, and 12 times sub-bands respectively, as shown in Figure~\ref{subfig:samples}. For comparison, we only randomly select once but double the size of CSI sample directly. We observe that augmenting the number of random subband selections is not as helpful as increasing the diversity of CSI samples. The reason behind this may be that whether we select 4, 8, or 12 subbands, they inherently convey the same underlying information about the multipath profile. Conversely, an increase in CSI samples results in a broader and more varied depiction of the multipath scenario.

Next, we investigate the impact of the number of time steps on the diffusion process. This parameter directly influences the noise scheduler table and thus the noise waveform. A higher time step adds more noise, meanwhile, it also necessitates a longer computation time. We vary the time steps from 10 to 100. The findings, depicted in Figure~\ref{subfig:ts}, indicate that beyond 50 time steps, increasing the level of noise does not significantly improve the system's performance.

We also study the impact of batch sizes, ranging from 64 to 2024 with a doubling interval. Figure~\ref{subfig:bs} shows that MSE decreases with increasing batch size up to 1024, but rises noticeably beyond this point.

In addition, since the size of the input CSI varies due to different protocols and bandwidth, \networkname chooses to zero-pad the input if it is too small. We also investigate the impact of the zero-pad operation versus the bilinear interpolation operation. The MSE results are 0.224967 and 0.226572, respectively. The results do not show much difference. Although a zero-pad operation on the input is believed to cause the network to easily get stuck in local minima, \networkname does not follow traditional deep learning practices to reconstruct the input; therefore, the zero-pad here is more like a mask for the network. In summary, we configure the randomly selected sub-band augment as 4 times, time steps as 50, and batch size as 1024. All results later in the paper are based on this setting.

\section{Case Studies: eCSI-based Sensing}
\label{sec:case_study}

We conduct two case studies to demonstrate the effectiveness of eCSI for finer-grained WiFi sensing. 

\subsection{Case I: Localization }

\label{sec:c1}
WiFi indoor localization is the most important sensing application because it compensates for the failure of GPS indoors. Time of Flight (ToF) is a crucial factor in enabling indoor localization and serves as the gold standard for calculating ranging error. However, obtaining the actual ToF information from the CSI to form distance estimates is challenging due to random hardware offsets, which do not accurately capture the true ToF of the signal from the target to the AP. As a result, many efforts have focused on increasing ToF accuracy by eliminating hardware offsets~\cite{xiong2015tonetrack}, which is beyond the system's scope. 

\begin{figure}[t]
\centering
\includegraphics[width=0.98\linewidth]{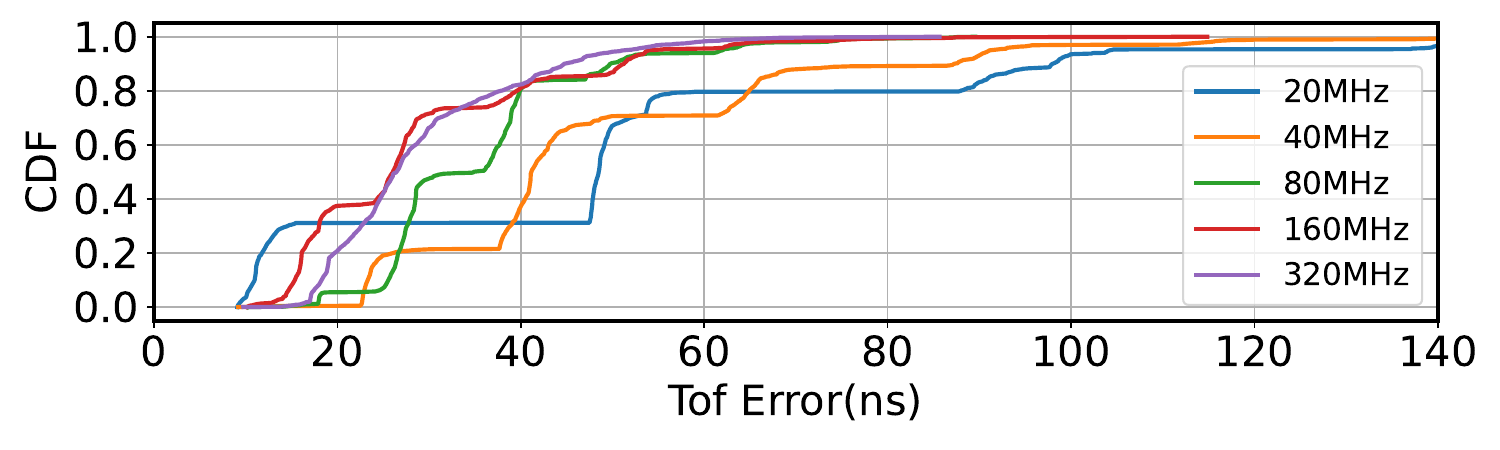}
 %  \vspace{-0.9\baselineskip}
    \caption{ToF errors in DLoc dataset, where 20,40,80MHz CSI is physical CSI from the dataset, while 160MHz and 320MHz are eCSI based on 80MHz physical CSI.} \label{subfig:tof}
  %  \vspace{-5mm}
\end{figure}

We aim to investigate whether the increased resolvability of distinguishable propagation paths can result in reduced ToF error. We conduct our evaluation on the DLoc dataset \cite{ayyalasomayajula2020deep}, which retains the ground truth of ToF error. This dataset is collected from two different environments, one in a multipath-rich and non-line-of-sight (NLOS) environment spanning 1500 square feet, and the other in a line-of-sight (LOS) environment. 
% Each environment includes multiple APs, with 4 and 3 respectively. 
The dataset contains CSI samples from 20MHz, 40MHz, and 80MHz. The ground truth ToF is measured using a robot platform equipped with LIDAR, camera, and odometry. 
% The simplest method is utilized to compute the ToF, and the first peak is selected to calculate the ToF error. 

We adopt a simple approach to calculate the ToF from CIR by finding the first dominant peak. 
The results are depicted in Figure~\ref{subfig:tof}, with the 20, 40, and 80MHz errors derived directly from the data, and the 160MHz and 360MHz errors obtained from the 80MHz physical CSI.  
As the bandwidth of CSI increases from 20MHz to 80MHz and to eCSI with 320MHz, the mean ToF error generally decreases, dropping from around 50.03ns to 35.34ns and 29.29ns.
As the bandwidth increases, there is an observed increase in the error for the 90th percentile, rising from 90.03ns at 20MHz to 49.77ns at 80MHz, and further to 51.26ns at 160MHz eCSI. However, there is a slight decrease at 320MHz eCSI, where the error rate for the 90th percentile reaches 44.73ns, falling between the error rates at 160MHz and 80MHz.
These findings suggest that higher frequency ToF systems may provide improved average performance, but the distribution of errors is more complex and may require further investigation to understand the factors contributing to the outliers observed in the data.
\vspace{-0.9\baselineskip}
% 90th percentile of tof error of 20MHz is 98.033343398338
% mean of tof error of 40MHz is 48.35697611885289
% 90th percentile of tof error of 40MHz is 87.27333000467104
% mean of tof error of 80MHz is 665784765258
% 90th percentile of tof error of 80MHz is 49.77201942899326
% mean of tof error of 160MHz is 28.753554719973685
% 90th percentile of tof error of 160MHz is 51.2649965935228
% mean of tof error of 320MHz is 29.290807922278663
% 90th percentile of tof error of 320MHz is 44.72865897093884

\subsection{Case II: Multi-person Breath Estimation}
\label{sec:c2}

Breathe estimation is one of the killer applications of WiFi sensing. 
%application, facilitating non-intrusive monitoring for smart healthcare applications.
Nevertheless, the limited bandwidth of WiFi fails to offer sufficient range resolution, which poses a significant challenge for multi-user breathing rate estimation. 
Existing solutions either require an extremely high frame rate (1000 frames/s) with calibration~\cite{zeng2020multisense} or the presence of a near-field signal from another proximity IoT device~\cite{hu2023muse} to enable multi-person breath estimation. As \sysname is able to obtain instant eCSI with improved multipath resolution, we are interested in evaluating how it could enable multi-person breathing rate estimation in practice. 

We carried out the case study in a square room with a size of 7m x 7m and set a pair of transceivers 1.6m apart. Three subjects sat at varying vertical distances of 1m, 2m, and 4m from the line of sight (LOS) between the WiFi Tx and Rx, as denoted by the red crosses in Figure~\ref{subfig:layout}. All three subjects breathed at the same respiration rate of 15 bpm but took turns to hold their breath. 
\begin{figure}[th!]
%\vspace{-0.7\baselineskip}
       \subfloat[Experimental setup]{
    \label{subfig:layout}
  \includegraphics[width=0.42\linewidth]{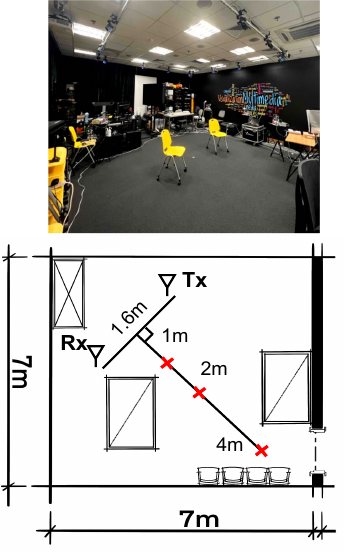}}
  \quad
    \subfloat[Multi-person breath estimation]{
    \label{subfig:bpm}
  \includegraphics[width=0.5\linewidth]{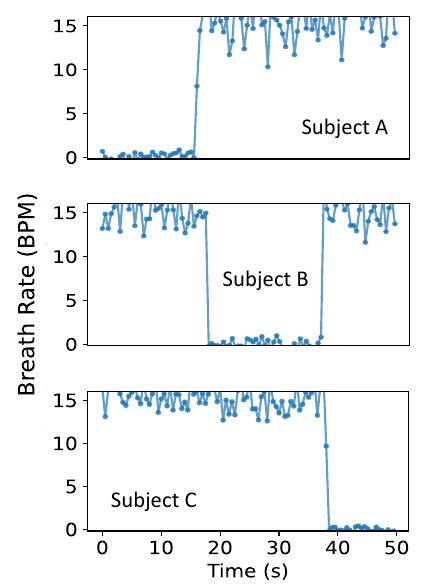}}
 
 %  \vspace{-0.7\baselineskip}
    \caption{Multi-person breath estimation: (a) Experimental setup: three people breathing at the same respiratory rate in the position of the red cross as denoted in the floor plan. (b) Breathe estimation for different subjects.} %\label{subfig:layout}
  %  \vspace{-6mm}
\end{figure}
\begin{figure}[h!]
\vspace{-0.9\baselineskip}
       \subfloat[The physical CIR of 20MHz]{
    \label{subfig:cirgt}
  \includegraphics[width=0.45\linewidth]{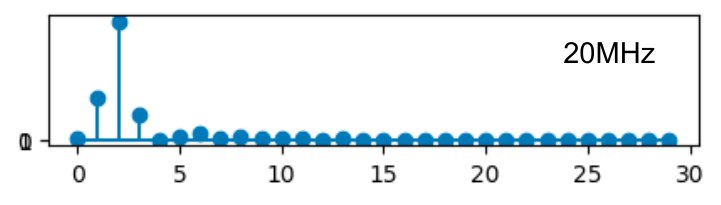}}
  \quad
    \subfloat[The CIR from \sysname]{
    \label{subfig:bigcir}
  \includegraphics[width=0.45\linewidth]{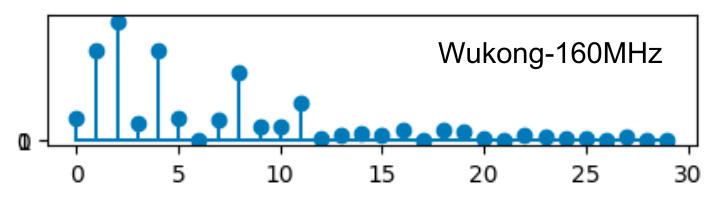}}
 % \quad
   % \vspace{-0.9\baselineskip}
   \caption{CIR of pCSI v.s. eCSI}\label{f:cirr}
   \vspace{-6mm}
\end{figure}
To estimate the respiratory rate, the CSI is first converted to CIR, and then a peak identification algorithm is carried out to identify the dominant paths. We omit the first peak because it reflects the direct path and then use the remaining three peaks as paths to identify the individual. We then conduct STFT of that CIR path with an 8s window size to obtain breathing rates. As shown in Figure~\ref{subfig:bpm}, with the \sysname-enabled NWB eCSI, we can estimate the breathing rates of multiple persons accurately and robustly. 
In comparison, using the physical CSI of 20 MHz fails to separate the multiple users, as illustrated in Figure~\ref{subfig:cirgt}.
The direct path and the reflected paths are inseparable in the CIR due to the limited resolution. 
As depicted in Figure~\ref{subfig:bigcir}, we can observe that the eCIR, the counterpart of eCSI, captures finer-grained reflection paths from different individuals, enabling the breathing rate estimation of each user separately. 
We the increased multipath resolution, at zero cost, we believe \sysname would enable a wide range of applications beyond breathing monitoring.

\section{Discussion}
\label{sec:discussion}

Wukong conceptualizes and pioneers the idea of Neuro-Wideband (NWB) WiFi sensing, introducing a new paradigm for expanding physical bandwidth using a single narrowband CSI measurement. While our work takes an important first step, several limitations remain and offer opportunities for future enhancement.

First, due to the lack of tools capable of directly measuring true ultra-wideband (UWB) WiFi CSI on commodity hardware, we rely on UWB (e.g., DW1000) data to validate extrapolation at higher bandwidths ($\geq$500 MHz). Although Wukong is designed to extrapolate to arbitrarily large bandwidths in principle, our current validation is physically grounded only up to 160 MHz on WiFi platforms. We thus conservatively refrain from claiming full UWB capability, and leave the integration of true WiFi-based UWB ground truth as important future work.

Second, while Wukong significantly improves time-of-flight (ToF) resolution and multipath separation, ToF accuracy remains susceptible to phase calibration errors and synchronization offsets, which are prevalent in real-world deployments~\cite{xiong2015tonetrack}. Future efforts will focus on jointly modeling phase offsets and hardware imperfections during eCSI inference to enhance localization accuracy under practical conditions.

Third, although we demonstrate the utility of Wukong-generated eCSI in localization and multi-user vital sign monitoring, we believe NWB sensing has broader implications. The underlying concept of frequency-aware extrapolation may extend naturally to other domains such as acoustic sensing, sub-6GHz radar, or mmWave communication. We plan to explore these cross-domain applications in future research.

Finally, while our training procedure enables self-supervised learning without requiring wideband ground-truth labels, model performance can benefit from greater diversity in environments and channel conditions. Scaling up the training corpus across domains (e.g., outdoor, mobile, NLOS) will further enhance generalization and robustness.

In summary, we view Wukong as a foundational step toward enabling high-resolution, real-time WiFi sensing using only standard hardware. We envision this work opening the door to a new era of eCSI-based sensing across diverse environments and applications.

\section{Related Works}
\label{sec:related_works}
\subsection{Improving Resolution for WiFi sensing}
The sensing resolution of WiFi is dependent on the number of antennas (for AoA) and the transmission frequency bandwidth (for ToF), respectively. Increasing the number of antennas and radios results in additional hardware costs~\cite{xiong2013arraytrack}. Recent research has introduced concepts such as virtual antenna arrays~\cite{kotaru2015spotfi,wu2019rf} and larger bandwidth through channel hopping~\cite{vasisht2016decimeter}. However, these methods require a static channel period and may impose constraints on ongoing data communication, which NWB aims to avoid. In addition, some works attempt to jointly estimate AoA and ToF~\cite{xie2019md,kotaru2015spotfi}. Yet, these methods often require calibration and are inherently limited in their design, resulting in high computational complexity and less scalability. In the most recent work, researchers use fast channel hopping to sample 20 CSIs and get finger-grain AoA-ToF tuples~\cite{li2024uwb}.   In contrast to previous approaches, NWB innovatively explores and exploits the CSI without changing existing hardware and communication configurations.

\subsection{Deep Generation-based RF Sensing}
In recent years, many generative models used in the field of RF sensing~\cite{yang2022deep},including localization~\cite{rizk2019effectiveness}, liquid classification~\cite{ha2020food}, gait recognition~\cite{yang2023xgait}, \textit{etc}.
Nevertheless, these models rely on environment-dependent features and generate samples in a specific environment, or even need transfer learning to apply the proposed method for othedownstream tasks~\cite{ha2020food}. The generalization ability to various unseen environments needs further exploration as well. 
These shortcomings greatly constrain their application in real-world scenarios. In addition, they are designed for a specific task and do not generalize to other sensing applications.
On the other hand, recently there has been a new trend of research work on directly generating signals. The recently proposed NeRF2 [68] learns the properties of the signal propagation
space based on NeRF. However, NeRF2 is limited to specific static scenarios and degrades for dynamic real-world scenarios. 
RF-Diffusion~\cite{chi2024rf} is the first work that uses time-frequency diffusion for both WiFi and mmWave signals, which boosts the accuracy of wireless sensing systems. Despite the success of generating RF data samples, none of them can improve the resolution of RF signals, not to mention the generalize towards distinct environments, devices or even protocols. 
Unlike previous endeavors, \sysname does not depend on the generative model to directly generate RF signals. Instead, it creatively utilizes the generative model to produce noise for self-conditioned extrapolation learning. As a result, for the first time, one can enhance RF sensing from the perspective of increasing bandwidth, enabling high-quality RF sensing across various applications.

\head{Channel Estimation}
Our work overlaps with the scope of physical channel estimation, a long-standing challenge in wireless communication. Traditional feedback-based approaches often suffer from excessive overhead~\cite{guo2022overview,guo2024deep}. More recent efforts~\cite{vasisht_eliminating_2016,bakshi2019fast,liu_fire_2021} attempt to predict unmeasured CSI (e.g., downlink) from measured CSI (e.g., uplink) based on assumptions like spatial reciprocity.
R2F2~\cite{bakshi2019fast} estimates multipath parameters via maximum likelihood but requires multiple-antenna devices and scales poorly with increasing path complexity. OptML and similar methods focus on mapping between different links or timestamps, but are limited by discontinuous frequency coverage and strong assumptions about channel stationarity. Additionally, FIRE~\cite{liu_fire_2021} highlights that even with deep learning, these methods fail to meet MIMO-grade accuracy.
More importantly, these methods are confined to predicting CSI at a different frequency and a different time, whereas our goal is to expand the continuous frequency band \textbf{at the same moment in time}, which is fundamentally different. While these methods could theoretically stitch predicted bands to form a quasi-wideband CSI, the resulting data is typically fragmented and lacks continuity.
In contrast, we formulate a new research problem: self-conditioned CSI extrapolation, where the goal is to expand the frequency coverage of a single physical CSI measurement without extra hardware or measurements. To the best of our knowledge, this is the first work to define and solve this problem systematically via a novel deep learning framework.

\section{Conclusion}
\label{sec:conclusion}

This paper presents \term{Neuro-Wideband (NWB)}, a novel paradigm for realizing wideband wireless sensing without requiring specialized hardware or multi-channel measurements. The core idea of NWB is to \emph{extrapolate} a physical CSI sample—measured using standard bandwidths on commodity WiFi—into an expanded, continuous-bandwidth representation, termed \term{eCSI}.

We introduce \sysname, a deep learning middleware that enables NWB by solving a self-conditioned extrapolation problem. By transferring sample-specific multipath characteristics embedded in the input CSI to unmeasured frequency components, \sysname generates high-fidelity eCSI across extended and continuous bandwidths. Importantly, it does so robustly across diverse environments and hardware platforms, without requiring retraining or calibration.

Extensive experiments demonstrate \sysname's ability to faithfully extrapolate eCSI, with physical validation up to 160 MHz. Case studies on localization and multi-person vital sign monitoring further highlight the practical benefits of eCSI-based sensing.

We believe NWB opens up a promising direction for bandwidth-constrained WiFi sensing, and that \sysname takes an important first step toward unlocking high-resolution, real-time sensing capabilities on commodity wireless systems.